\documentclass[a4paper,10pt]{scrartcl}

\usepackage{amssymb}
\usepackage[pdftex]{graphicx}

\usepackage[utf8x]{inputenc}
\usepackage[T1]{fontenc}
\usepackage{algpseudocode}  % used for algorithms in pseudocode
\usepackage{amsmath}
\usepackage{algorithm}
\usepackage{xspace}
\usepackage{amsmath,amssymb,amstext,amsthm}

\newenvironment{proofof}[1]{\bigskip \noindent {\bf Proof of #1:} }
{\qed\par\vskip 4mm\par}

\usepackage{tikz}	% figures and TODO boxes are generated with tikz

\newcommand{\nodesize}{6pt}

\newcommand{\labelednode}[4]{
        \node[node] (#2) at (#1){};
        \node[empty,yshift=1em] at (#1) {#3};
}
\newcommand{\xshift}{2}

\usetikzlibrary{fit}
\usetikzlibrary{calc}
\usetikzlibrary{snakes}

\tikzstyle{every picture}+=[
node/.style={circle, minimum size=\nodesize, draw=black, inner sep=0,fill=white},       % print a circle
empty/.style={draw=none,fill=none},
curvedLine/.style={decorate,decoration={snake,amplitude=.5mm}}
]

\newcommand{\N}{\mathbb{N}}

\newcommand{\mydef}[1]{\emph{#1}\xspace}  % own definition macro
\newcommand{\key}[1]{\ensuremath{\textsf{key}(#1)}\xspace}    % key of a data item
\newcommand{\numblocks}{\ensuremath{\Lambda}\xspace}  % number of blocks of an address
\newcommand{\bit}[2]{\ensuremath{\textsf{bit}_{#1}(#2)}\xspace} % #1-th bit of a data item #2
\newcommand{\bucket}[1]{\ensuremath{B}_{#1}\xspace} % bucket
\newcommand{\bucketServer}[2]{\ensuremath{B_{#1,#2}}\xspace}  % stored data items of server #2 in bucket #1
\newcommand{\requestServer}[2]{\ensuremath{D_{#1,#2}}\xspace} % data items server #2 has requests for in round #1
\newcommand{\rep}[1]{\ensuremath{s(#1)}\xspace} % bucket

\newcommand{\zerochild}{$0$\text{-}child\xspace}  % $0$-child
\newcommand{\onechild}{$1$\text{-}child\xspace} % $1$-child
\newcommand{\zerochildarg}[1]{\zerochild(#1)\xspace}  % $0$-child
\newcommand{\onechildarg}[1]{\onechild(#1)\xspace}  % $1$-child

\newcommand{\fbucket}[2]{\ensuremath{bucket(#1,#2)}\xspace} % bucket computation function

\newcommand{\whp}{w.h.p.\xspace}  % abbreviation for ``with high probability''

\newcommand{\serverset}{\ensuremath{\mathcal{S}}\xspace}  % set of servers

\newcommand{\lookupreq}[1]{\ensuremath{\textsf{lookup}(#1)}\xspace} % lookup request for d
\newcommand{\writereq}[1]{\ensuremath{\textsf{write}(#1)}\xspace}   % write request for d
\newcommand{\datafrom}[1]{\ensuremath{\mathcal{D}(#1)}\xspace}  % set of data items with a write request in #1

\newcommand{\probenull}[0]{\textsf{probe($\cdot$)}\xspace}  % probe keyword, defines message type
\newcommand{\probe}[3]{\textsf{probe(\ensuremath{#1,#2,#3})}\xspace}  % probe keyword, defines message type -- example: probe(d,i,t)
\newcommand{\decodemsg}[3]{\textsf{decode(\ensuremath{#1,#2,#3})}\xspace} % decode keyword, defines message type  -- example: dedcode(d,i,t)
\newcommand{\decodemsgnull}[0]{\textsf{decode($\cdot$)}\xspace} % decode keyword, defines message type
\newcommand{\decodecheckmsg}[2]{\textsf{decodeCHECK(\ensuremath{#1,#2})}\xspace}  % decode keyword, defines message type
\newcommand{\decodecheckmsgnull}[0]{\textsf{decodeCHECK($\cdot$)}\xspace} % decode keyword, defines message type
\newcommand{\congmsg}[1]{\textsf{cong(\ensuremath{#1})}\xspace} % cong keyword, defines message type
\newcommand{\failmsg}[3]{\textsf{fail(\ensuremath{#1,#2,#3})}\xspace} % fail keyword, defines message type --- example: fail(d,i,t)
\newcommand{\failmsgnull}[0]{\textsf{fail($\cdot$)}\xspace} % fail keyword, defines message type --- example: fail(d,i,t)
\newcommand{\notexistsmsg}[1]{\textsf{notexists(\ensuremath{#1})}\xspace} % notexist keyword, defines message type
\newcommand{\msgfull}[0]{\ensuremath{\textsf{full}}\xspace} % full keyword, defines message type
\newcommand{\msgpartly}[1]{\ensuremath{\textsf{partly(#1)}}\xspace} % full keyword, defines message type

\newcommand{\nodeatlvl}[3]{\ensuremath{s_{#1}^{(#2)}(#3)}\xspace} 
\newcommand{\serversinvolved}[3]{\ensuremath{\Gamma_{#1,#2}(#3)}\xspace}
\newcommand{\datap}[2]{\ensuremath{#1_{#2}}\xspace}  %% data item piece (currently: d_i for {d}{i})
\newcommand{\timestamp}[1]{\ensuremath{t_{#1}}\xspace}   %% timestamp (currently: t_d for {d})

\newcommand{\num}[1]{\ensuremath{\textsf{num}_{#1}}}  % number of specific data items (used in counting section)

\newcommand{\curphase}{\ensuremath{z}}
\newcommand{\subphase}{sub-phase\xspace}

\newcommand{\numbercrashedservers}{\ensuremath{\Theta(n^{1/\log\log{n}})}\xspace}

\newtheorem{theorem}{Theorem}[section]
\newtheorem{lemma}[theorem]{Lemma}
\newtheorem{claim}[theorem]{Claim}

\newtheorem{corollary}[theorem]{Corollary}
\newtheorem{definition}[theorem]{Definition}

\usepackage{url}

\title{\Large
\textrm{
RoBuSt: A Crash-Failure-Resistant \\
Distributed Storage System\footnote{This work was partially supported by the German Research Foundation (DFG) within the Collaborative Research Center ``On-The-Fly Computing'' (SFB 901) and by the EU within FET project MULTIPLEX under contract no. 317532.}\\
\small Revised full version
}}

\author{
\normalsize Martina Eikel
\and\normalsize  Christian Scheideler
\and\normalsize  Alexander Setzer
}

\date{
	\normalsize
	University Paderborn, Germany
}

\begin{document}

\maketitle

\begin{abstract}
In this work we present the first distributed storage system that is provably
robust against crash failures issued by an adaptive adversary, i.e., for each
batch of requests the adversary can decide based on the entire system state
which servers will be unavailable for that batch of requests.
Despite up to $\gamma n^{1/\log\log n}$ crashed servers, with $\gamma>0$
constant and $n$ denoting the number of servers, our system can
correctly process any batch of lookup and write requests (with at most a
polylogarithmic number of requests issued at each non-crashed server) in at
most a polylogarithmic number of communication rounds, with at most
polylogarithmic time and work at each server and only a logarithmic storage
overhead.

Our system is based on previous work by Eikel and Scheideler (SPAA 2013), who
presented IRIS, a distributed information system that is provably robust
against the same kind of crash failures. However, IRIS is only able to serve
lookup requests.
Handling both lookup and write requests has turned out to require major changes in the design
of IRIS.
\end{abstract}

\section{Introduction}

%%%%%%%%%%%%%%%%%%%%%%%%%%%%%%%%%%%%%%%%%%%%%%%%%%%%%%%%%%%%%%%%%%%%
% Crash failures
One of the main challenges of a distributed system is that it is able to work
correctly even if parts of the system fail to work. If a server experiences a
\emph{crash failure} it becomes unavailable to the other servers, i.e., it
does not issue or respond to requests any more. Crash failures can be
temporary or permanent, and if it is temporary, a server may either be back to
its state when it crashed, or it may have lost all of its state. We will focus
on crash failures where, whenever a server becomes available again, it is back
to its state when it crashed. This is a reasonable assumption since for
commercial servers it is extremely rare that their state cannot be recovered.
However, a temporary unavailability is not that uncommon and can have many
causes such as maintenance work, hardware or software glitches, or
denial-of-service attacks. Especially denial-of-service attacks can be a
serious threat because they are normally unpredictable, hard to prevent and
they can cause the unavailability of a server for an extended period of time.

Predominant approaches in information and storage systems to deal with the threat of crash
failures are to use redundancy: information that is replicated among multiple
machines is likely to remain accessible even if some servers are unavailable.
Unfortunately, in systems that consist of thousands of servers a complete
replication of the data over all servers is not feasible. Hence, one needs to
find an appropriate tradeoff between the amount of redundancy and the number
of crashed servers the system can handle. One can easily show that if
$\Theta(\log n)$ copies of a data item are placed randomly among $n$ servers,
and these random positions are not known to the adversary, then any strategy
of the adversary that blocks half of the servers will not block all of the
copies, with high probability\footnote{``With high probability'', or short,
``w.h.p.'', means a probability of at least $1-1/n^c$ where the constant $c$
can be made arbitrarily large.}. The situation is completely different,
however, when considering an adaptive adversary, i.e., someone who has
complete knowledge about the system.

In a previous work, Eikel and Scheideler \cite{iris} presented a distributed
information system, called IRIS, that just needs a constant storage redundancy
in order to be robust against an adaptive adversary that can crash up to
$\numbercrashedservers$ servers. Unfortunately, the system lacks the important
ability to handle write requests, i.e., to add, remove and update data items.
This work solves this problem.

%%%%%%%%%%%%%%%%%%%%%%%%%%%%%%%%%%%%%%%%%%%%%%%%%%%%%%%%%%%%%%%%%%%%
\subsection{Model and Preliminaries}

We assume that the storage system consists of a static set $\serverset =
\{s_1, \dots, s_n\}$ of $n$ reliable servers of identical type. The servers
are responsible for storing the data as well as handling the user requests. We
assume that all data items are of the same size, and that any data item $d$ is
uniquely identified by a key $key(d)$. The universe of all possible keys is
denoted by $U$, and $m := |U|$ is assumed to be polynomial in $n$.
Furthermore, we assume that the size of the data items is at least $\Omega(\log n\log m)$. There are two types of user
requests: \lookupreq{k} for $k \in U$, and \writereq{k,d} for $k \in U$ and a
data item $d$. The user can issue a request by sending it to one of the
servers in $\serverset$. Given a \lookupreq{k} request, the system is supposed
to either return the data item $d$ with $key(d)=k$, or to return \textsf{NULL}
if no such data item exists. Given a \writereq{k,d} request, the system is
supposed to store data item $d$ with key $k$ such that subsequent
\lookupreq{k} requests can be answered correctly. Note that with a
\writereq{\cdot} request the user can also update or remove data.

Every server knows about all other servers and can therefore directly communicate with any one of them.
This does not endanger scalability since millions of IP addresses can easily be stored in main memory in any reasonable computer today and we assume the set of servers to be static.
We use the standard synchronous message passing model for the communication between the servers.
That is, time proceeds in synchronized {\em communication rounds}, or simply {\em rounds}, and in each round each server first receives all messages sent to it in the previous round, processes all of them, and then sends out all messages that it wants to send out in this round.
Note that assuming the synchronous model is just a simplification and that our protocols only require the message delays to be bounded.
In addition, we use the synchronous model because describing all protocols in an asynchronous setting would significantly blow up the construction and would hide the main innovations behind our system.
We assume that the time needed for internal computations is negligible, which is reasonable as the operations in the protocols we describe are simple enough to satisfy this property.

For the crash failures, we assume a \emph{batch-based} adaptive adversary.
This means the following:
We assume that time is divided into \mydef{periods} consisting of a polylogarithmic number of rounds.
The adversary has complete knowledge of the current system, but cannot predict the (future) random choices of the system.
Based on his knowledge, he can select an arbitrary set of $O(n^{1/\log\log n})$ servers to be crashed.
A server that is crashed will not send any message nor react to messages sent from other servers.
We assume that the servers have a failure detector that allows them to determine whether a server is crashed so that statements like ``if server $i$ is crashed then \ldots'' are allowed in the protocol.
Note that assuming bounded message delays, failure detection can simply be implemented using timeouts.
After that, the adversary may issue an arbitrary collection of requests to the
system by sending up to $\omega\in\mathbb{N}$ \lookupreq{\cdot} requests and
up to $\omega$ \writereq{\cdot} requests to each server. In order to keep the
presentation of RoBuSt as clear as possible, throughout this work we assume
$\omega = 1$. RoBuSt can in principle handle arbitrary values of $\omega$, but
in that case the bound on the work required by each server for serving all
requests must be multiplied with $\omega$.\footnote{Note that our system would
not be able to answer all requests with at most polylogarithmic work if
$\omega > polylog(n)$, but this would trivially hold for any storage system.}
There are no further limitations, i.e., the keys selected by the adversary may
or may not be associated with data items stored in the system, and the
adversary is also allowed to issue multiple lookup requests for the same key.
The task of the system is to correctly handle \emph{all} of these requests.
We assume that any period is long enough such that the system has enough time
to perform all necessary computations and to answer all requests. After any
period, the adversary may select a different set of $\numbercrashedservers$
servers to be crashed. We assume that the set of crashed servers does not
change during a fixed period, which is why we use the notion of a batch-based
adaptive adversary. Of course, allowing crash failures at arbitrary times
would make the model much stronger, yet it would significantly complicate the
system design, which is why we leave this to future research. Note that we
assume links between intact (i.e., non-crashed) servers to be reliable.
Unreliable links can be dealt with using, for example, at-least-once delivery
or error correction strategies, which are out of scope for our design since it
is already complex enough.

In order to measure the quality of the storage system, we introduce the following notation.
A storage strategy is said to have a {\em redundancy} of $r$ if $r$ times more storage (including any control storage) is used for the data than storing the plain data.
We call a storage system {\em scalable} if its redundancy is at most polylog($n$), {\em efficient} if any collection of lookup and write requests specified by the adversary can be processed correctly in at most polylog($n$) many communication rounds in which every server sends and receives at most polylog($n$) many messages of at most polylog($n$) size, and {\em robust} if any collection of lookup and write requests specified by the adversary can be processed correctly even if a set of up to $\numbercrashedservers$ servers specified by the adversary crash.

%%%%%%%%%%%%%%%%%%%%%%%%%%%%%%%%%%%%%%%%%%%%%%%%%%%%%%%%%%%%%%%%%%%%
\subsection{Related Work}
Over the past years, distributed storage systems have gained a lot of importance.
Popular examples include the storage solutions offered by Google, Apple, or Amazon.
Since availability and retrievability of the stored data is a key aspect of
distributed storage systems, these systems should be able to work correctly
despite common failures. Often failures in distributed systems are divided
into the following types \cite{christian-1991}: crash failures, omission
failures, timing failures, and Byzantine failures. In crash failures the
affected component (for instance a server) completely stops working.
In receive (send) omission failures the affected component cannot receive
(send) any further messages. A timing failure leads a component to not respond
within a specified time interval. In case of a Byzantine failure, the affected
component may react in an arbitrary, even malicious manner.

This work focuses on crash failures. Many works dealing with crash failures in
distributed systems focus on crash failure recovery and crash failure
detection \cite{sistla-1989,leners-2011,gupta-2001}. But to the best of our knowledge, no
previous work has considered how to secure a distributed
storage system against many (e.g., more than a polylogarithmic number)
simultaneous crash failures controlled by an adaptive adversary while using
only polylogarithmic work, time and redundancy. That is, we do not seek to
prevent failures or attacks, but rather focus on how to maintain a good
availability and performance even in spite of them. Our system is based on the
distributed hash table (DHT) paradigm
(e.g.,~\cite{pagoda,pastry,skipnet,can,techchord}), with the additional twist
of using coding and arranging the used DHTs in an appropriate structure.
Various systems based on DHTs that are resistant against Denial-of-Service
(DoS) attacks (which represent a special type of crash failures) have already
been proposed \cite{out6,out8,out13}. But these do not work for adaptive
adversaries. The first DHTs that are robust against past insider crash
failures were proposed in \cite{AS-2007,scheideler2009}, where a past insider
only has complete knowledge of the information system up to some {\em past}
time point $t_0$. For this kind of insider, it is possible to design an
information system so that any information that was inserted or last updated
{\em after} $t_0$ is safe against crash failures
\cite{AS-2007,scheideler2009}. But the constructions proposed in these papers
would not work at all for a current insider because they are heavily based on
randomization to ensure unpredictability.
%
% IRIS
Eikel and Scheideler were the first to present a distributed information
system, called IRIS, that is provably robust even against a current insider
that crashes up to $\Theta(n^{1/\log\log{n}})$ servers.
The authors showed that IRIS can correctly answer any set of lookup requests
(with one request per server that is not crashed) with
polylogarithmic time and work at each server and only a constant redundancy.
Still it remained open whether it is possible to design a distributed storage
system that can efficiently handle lookup and write requests under the
presence of crash failures. We answer this question positively by proposing
such a system.

%%%%%%%%%%%%%%%%%%%%%%%%%%%%%%%%%%%%%%%%%%%%%%%%%%%%%%%%%%%%%%%%%%%%
\subsection{Our Contribution}
We present the first scalable distributed storage system, called \mydef{Robust Bucket Storage (in short RoBuSt)}, that is provably robust against adaptive crash failures and that supports both lookup and write requests.
Concretely, we allow the adversary to have complete knowledge about the storage system and to have the power to crash any set of $\gamma n^{1/\log\log n}$ servers, for $\gamma > 0$ constant.
The task of the system is to serve any collection of lookup and write requests in an efficient way despite the crash failures.

RoBuSt expands some of the ideas in IRIS, a distributed storage system that we proposed in SPAA 2013 \cite{iris}.
The system presented in this work tolerates a number of crashed servers that is similar to the number of servers blocked by a DoS attack that the Basic IRIS version can tolerate and achieves comparable efficiency bounds (up to a logarithmic factor).
In contrast to IRIS, which can only handle lookup requests, RoBuSt is able to additionally handle write requests.
Although in the lookup protocol we are able to adapt some of the underlying ideas of IRIS, adding the write functionality required significant changes in the whole structure.
To simplify the description for readers who are familiar with IRIS, we try to re-use terminology whenever there are similarities (e.g., Probing Stage, Decoding Stage).

One aspect is that IRIS organizes data into layers of $n$ data items each, and each layer is encoded separately using distributed coding that involves all $n$ servers.
This means that whenever a data item needs an update, all $n$ servers have to update their information for the corresponding layer.
Since we allow any set of write requests, it may happen that every write request involves a different layer, which would create an enormous update work.
To solve this issue, in RoBuSt we store the data items in so-called buckets that are organized in a binary tree.
For each data item, there are a logarithmic number of buckets that are a potential storage location for the data item.
For a data item there may exist different versions of it in different buckets.
But our system ensures that the highest bucket (i.e., the bucket with minimum distance to the root in the underlying binary tree over the buckets) that contains a version of the data item always holds the most recent version.

Furthermore, IRIS uses a fixed set of hash functions to specify anchor locations for the data so that afterwards lookup requests can be served efficiently despite an adversarial DoS attack.
However, using fixed hash functions in RoBust would enable the adversary to annul the fair distribution of data in a bucket.
Therefore, RoBuSt chooses new, random hash functions whenever write requests have to be served.

Another complication is the fact that a server may not know whether its information is up-to-date. This is because at the time when write requests were executed that required an update in that server, the server might have been crashed.
Our organization of the data and our protocols ensure that any server that answers a request always returns the most recent version of a data item.

Nevertheless, RoBuSt makes sure that all data can still be efficiently found while the storage overhead is at most a logarithmic factor.

\begin{theorem}\label{theorem:main}
RoBuSt is a scalable and efficient distributed storage system that only needs a logarithmic redundancy to protect itself against batch-based adaptive crash failures in which up to $\gamma\cdot n^{1/\log\log n}$ servers can crash for a constant $\gamma>0$, \whp
\end{theorem}

%%%%%%%%%%%%%%%%%%%%%%%%%%%%%%%%%%%%%%%%%%%%%%%%%%%%%%%%%%%%%%%%%%%%
\section{Underlying Datastructure}\label{sec:underlying-ds}

In the following, we assume keys are potentially from an address space of size at most $n^p$, i.e., we need $\Lambda:=p\log n$ bits for each address.
We introduce the following definitions:
For a data item $d$, denote the \mydef{address} of $d$ by $\key{d} = d_{p\log n-1}\ldots d_1 d_0\in\{0,1\}^{p\log n}$ and let $\bit{d}{i} := d_i$.

Our data structure is based on a binary tree with $\numblocks+1$ levels, so-called \mydef{zones}.
We denote the nodes of each zone as \mydef{buckets} where each bucket will hold a set of data items.
The internal storage strategy of the buckets is described in Section~\ref{sec:bucket-storage-strategy}.
Zone 0 consists of a single bucket, bucket $\bucket{\varepsilon}$.
Each bucket $B$ that is not in zone $\Lambda$ has two children, denoted by $\zerochildarg{B}$ and $\onechildarg{B}$.
For each data item $d$ there is not only a single possible bucket in which to store $d$
but there are $\numblocks + 1$ possible buckets for $d$, one in each zone.
Bucket $\bucket{\varepsilon}$ may hold any data item.
Any data item $d$ that may belong to bucket $B$ in zone $\ell$, may also belong to $\zerochildarg{B}$ if $\bit{d}{\ell} = 0$ or to $\onechildarg{B}$ if $\bit{d}{\ell} = 1$.
In the following, let $\mathcal{B}$ be the set of all buckets and let $\fbucket{z}{d}: \{0,\dots,\numblocks\} \times U \rightarrow \mathcal{B}$ be a function that returns the unique possible bucket of a data item $d$ at zone $z$. 
Initially, a bucket does not contain any data.
During the runtime of the system the following invariant is satisfied:
Each bucket, excluding bucket $\bucket{\varepsilon}$, stores either $0$ or between $n$ and $2n$ data items.
Bucket $\bucket{\varepsilon}$ stores at most $2n$ data items.

%%%%%%%%%%%%%%%%%%%%%%%%%%%%%%%%%%%%%%%%%%%%%%%%%%%%%%%%%%%%%%%%%%%%
\subsection{Internal Storage Strategy of the Buckets}\label{sec:bucket-storage-strategy}

The idea of storing a set $D$ of data items into a bucket $B$ is to reuse the basic concepts of the storage strategy for individual layers from IRIS \cite{iris}. 
Roughly speaking this strategy works as follows:
In order to achieve the desired robustness, we first create $c \geq 18\log{m}$ \mydef{pieces} $d_1,\ldots,d_c$ for each data item $d\in D$ using Reed Solomon coding.
Using $c$ hash functions chosen uniformly and independently at random, these pieces are then mapped to servers.
Finally, all these pieces are encoded with each other, such that at the end each intact server holds for each piece some parity information resulting from the encoding process.
Besides encoded data pieces each bucket $B$ additionally stores $c$ hash functions and a timestamp $t(B)$. The timestamp is used to handle out-dated information a server might hold if it has crashed in a previous period in which write requests were served.

%%%%%%%%%%%%%%%%%%%%%%%%%%%%%%%%%%%%%%%%%%%%%%%%%%%%%%%%%%%%%%%%%%%%
In the following we roughly describe the coding strategy presented in \cite{iris}.
The coding strategy is a block-based distributed strategy that follows the topology of a $k$-ary butterfly as described in the following.
For $k\in\N$ we use the notation $[k]=\{0,\ldots,k-1\}$.

\begin{definition}
For any $d,k \in \N$, the {\em $d$-dimensional $k$-ary butterfly} $BF(k,d)$ is
a graph $G=(V_k,E)$ with node set 
$V_k=[d+1]\times [k]^d$ and edge set $E$
with
\begin{align*}
  E= \{ & \{(i,x),(i+1,(x_1,\ldots,x_i,b,x_{i+2},\ldots,x_d))\} \\
   & \mid x=(x_1,\ldots,x_d) \in [k]^d,~i\in[d], \mbox{ and } b \in [k] \}.
\end{align*}
A node $u$ of the form $(\ell,x)$ is said to \mydef{be on butterfly level} $\ell$ of $G$.
Furthermore, $LT(u)$ is the unique $k$-ary tree of nodes reached from $u$ when going downwards the butterfly (i.e., to nodes on butterfly levels $\ell'>\ell$) and $UT(u)$ is the unique $k$-ary tree of nodes reached from $u$ when going upwards the butterfly.
Moreover, for a node $u$ at level $\ell$, let $BF(u)$ be the unique $k$-ary sub-butterfly of dimension $\ell$ ranging from butterfly level 0 to $\ell$ in $BF(k,d)$ that contains $u$.
\end{definition}

A visualization of a $k$-ary butterfly is given in Figure~\ref{fig:bf-visual}.

\begin{figure*}[htp]
\centering
% !TEX root = ../iris.tex

\tikzstyle{every picture}+=[node/.style={circle, minimum size=5, draw=black, inner sep=0,fill=white}]       % print a circle

\begin{tikzpicture}
	\def\nodexdist{0.4}
	\def\nodeydist{0.8}
	\def\numnodeslevel0{27}
	
	\foreach \x in {1,2,...,\numnodeslevel0} {
		% nodes
		\node[node] (\x-0) at (\x*\nodexdist,0) {}; % level 0 nodes
		\node[node] (\x-1) at (\x*\nodexdist,-\nodeydist) {}; % level 1 nodes
		\node[node] (\x-2) at (\x*\nodexdist,-2*\nodeydist) {}; % level 2 nodes
		\node[node] (\x-3) at (\x*\nodexdist,-3*\nodeydist) {}; % level 3 nodes		
	}	
	
	% Level names
	\node[empty] at (-0.6, 0) {Level 0};
	\node[empty] at (-0.6, -\nodeydist) {Level 1};
	\node[empty] at (-0.6, -2*\nodeydist) {Level 2};
	\node[empty] at (-0.6, -3*\nodeydist) {Level 3};
	
	% node labels
	\node[empty,yshift=1em] at (1-0) {$000$};
	\node[empty,yshift=2em] at (2-0) {$001$};
	\node[empty,yshift=1em] at (3-0) {$002$};
	
	\node[empty,yshift=2em] at (4-0) {$010$};
	\node[empty,yshift=1em] at (5-0) {$011$};
	\node[empty,yshift=2em] at (6-0) {$012$};
	
	\node[empty,yshift=1em] at (7-0) {$020$};
	\node[empty,yshift=2em] at (8-0) {$021$};
	\node[empty,yshift=1em] at (9-0) {$022$};
	
	\node[empty,yshift=2em] at (10-0) {$100$};
	\node[empty,yshift=1em] at (11-0) {$101$};
	\node[empty,yshift=2em] at (12-0) {$102$};

	\node[empty,yshift=1em] at (13-0) {$110$};
	\node[empty,yshift=2em] at (14-0) {$111$};
	\node[empty,yshift=1em] at (15-0) {$112$};
	
	\node[empty,yshift=2em] at (16-0) {$120$};
	\node[empty,yshift=1em] at (17-0) {$121$};
	\node[empty,yshift=2em] at (18-0) {$122$};
	
	\node[empty,yshift=1em] at (19-0) {$200$};
	\node[empty,yshift=2em] at (20-0) {$201$};
	\node[empty,yshift=1em] at (21-0) {$202$};

	\node[empty,yshift=2em] at (22-0) {$210$};
	\node[empty,yshift=1em] at (23-0) {$211$};
	\node[empty,yshift=2em] at (24-0) {$212$};

	\node[empty,yshift=1em] at (25-0) {$220$};
	\node[empty,yshift=2em] at (26-0) {$221$};
	\node[empty,yshift=1em] at (27-0) {$222$};

	% level 0 edges	
	\foreach \x in {1,2,...,\numnodeslevel0} {
		\pgfmathtruncatemacro{\mod}{mod(\x-1,3)+1} 			
		\pgfmathtruncatemacro{\xplusone}{\x+1}
		\pgfmathtruncatemacro{\xplustwo}{\x+2}
		\pgfmathtruncatemacro{\xminusone}{\x-1}
		\pgfmathtruncatemacro{\xminustwo}{\x-2}
		
		\draw (\x-0) -- (\x-1);
		\ifthenelse{\mod = 2}{
			\draw (\x-0) -- (\xplusone-1);
			\draw (\x-0) -- (\xminusone-1);			
		}{}
		
		\ifthenelse{\mod = 1}{
			\draw (\x-0) -- (\xplusone-1);
			\draw (\x-0) -- (\xplustwo-1);
		}{}

		\ifthenelse{\mod = 3}{
			\draw (\x-0) -- (\xminusone-1);
			\draw (\x-0) -- (\xminustwo-1);
		}{}
	}
	
	% draw ultra thick lines
	\draw[ultra thick] (10-0)--(11-1);
	\draw[ultra thick] (11-0)--(11-1);
	\draw[ultra thick] (12-0)--(11-1);
	\draw[ultra thick] (13-0)--(14-1);
	\draw[ultra thick] (14-0)--(14-1);
	\draw[ultra thick] (15-0)--(14-1);
	\draw[ultra thick] (16-0)--(17-1);
	\draw[ultra thick] (17-0)--(17-1);
	\draw[ultra thick] (18-0)--(17-1);
	
	% draw sub butterfly box
	\draw[dashed] (10-0.west) -- (10-2.west);
	\draw[dashed] (10-2.south) -- (18-2.south);
	\draw[dashed] (18-2.east) -- (18-0.east);
	\draw[dashed] (18-0.north) -- (10-0.north);
	
	% level 1 edges
	\draw (5-1) -- (5-2);
	\draw (5-1) -- (2-2);
	\draw (5-1) -- (8-2);
	
	\draw (2-1) -- (2-2);
	\draw (2-1) -- (5-2);
	\draw (2-1) -- (8-2);
	
	\draw (8-1) -- (8-2);
	\draw (8-1) -- (5-2);
	\draw (8-1) -- (2-2);
	
	\draw (14-1) -- (11-2);
	\draw[ultra thick] (14-1) -- (14-2);
	\draw (14-1) -- (17-2);
	
	\draw (11-1) -- (11-2);
	\draw[ultra thick] (11-1) -- (14-2);
	\draw (11-1) -- (17-2);
	
	\draw (17-1) -- (11-2);
	\draw[ultra thick] (17-1) -- (14-2);
	\draw (17-1) -- (17-2);
	
	% level 2 edges
	\draw[ultra thick, dotted] (17-2) -- (8-3);
	\draw[ultra thick, dotted] (17-2) -- (17-3);
	\draw[ultra thick, dotted] (17-2) -- (26-3);
	
	\draw (8-2) -- (8-3);
	\draw (8-2) -- (17-3);
	\draw (8-2) -- (26-3);
	
	\draw (26-2) -- (8-3);
	\draw (26-2) -- (17-3);
	\draw (26-2) -- (26-3);
\end{tikzpicture}
\caption{Visualization of a $k$-ary butterfly $BF(k,d)$ for $k=d=3$. For a
better readability most of the edges from level two and three are omitted.
The dashed box denotes the sub-butterfly $BF((2,111))$. 
The thick solid lines in the dashed box denote the edges of $UT((2,111))$.
The thick dotted lines denote the edges of $LT((2,121))$. } 
\label{fig:bf-visual}
\end{figure*}
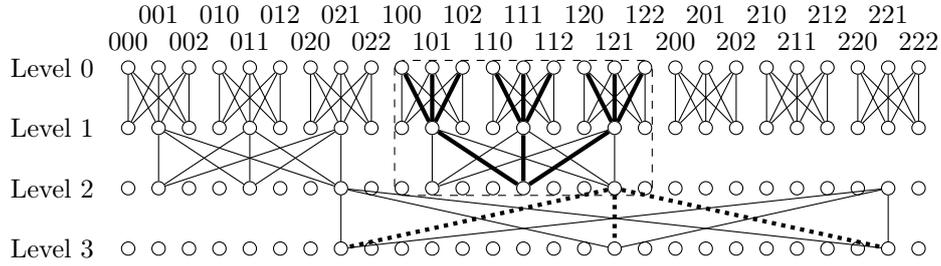

In the following let $BF(k,d)$ be a $k$-ary butterfly with $n=k^d$ and with server $s_i$, $i\in\{0,\ldots,n-1\}$, emulating the butterfly nodes $(0,i),\ldots,(d,i)$.
That is, whenever a butterfly node $(j,i)$, $j\in\{0,\ldots,d\}$ is supposed to perform an action or store data, this is done by server $s_i$.
We say a server $s$ is \mydef{connected via the $k$-ary butterfly} to another server $s'$, if there is an edge $(u,v)$ in the butterfly such that $u$ is emulated by $s$ and $v$ is emulated by $s'$.

While in IRIS each server holds $O(1)$ data pieces per layer, in our system each server holds $O(\log n)$ data pieces per bucket.
This is due to the fact that each bucket contains $O(n)$ data items and for each data item $c=\Theta(\log m)$ pieces are created and distributed evenly among the servers.
Hence, we simply concatenate the data pieces a server $s_i$ holds in a bucket $B$ and denote the resulting data \mydef{block} as $b_i$.

In order to encode the data blocks $b_0,\ldots,b_{n-1}$ assigned to the servers $s_0,\ldots,$ $s_{n-1}$ in bucket $B$, initially, $b_i$ is placed in node $(0,i)$ for every $i\in \{0,\ldots,n-1\}$.
Given that in butterfly level $\ell$ we have already assigned data blocks $d(\ell,x)$ to the nodes
$(\ell,x)$ we use the coding strategy presented in \cite{iris} to assign data blocks $d(\ell+1,x)$ to the nodes at butterfly level $\ell+1$.
The used coding strategy is based on some simple parity computations and ensures the following property: 
If at most one butterfly node $(\ell+1,x_j)$ from the set of nodes $\{(\ell+1,x_1),\ldots,(\ell+1,x_k)\}$ is crashed, then the information in the remaining nodes $(\ell+1,x_i)$, $i\in\{1,\ldots,k\}\backslash\{j\}$, suffices to recover $d(\ell,x_1),\ldots,d(\ell,x_k)$.
Furthermore, with Lemma~2.4 in \cite{iris} the storage amount of each server $s_i$, $i\in\{0,\ldots,n-1\}$, required for the encoding of a single bucket is upper bounded by $(1+e)z$, where $z$ denotes the maximum size of the data blocks stored at any server $s_j$, $j\in\{0,\ldots,n-1\}$.
Since there may exist outdated data items in the system, but for each level at most one, i.e. in total at most $\Lambda+1=O(\log n)$ many for each data item, the redundancy of our system increases to $O(\log n)$.

\begin{corollary}\label{corollary:storage-overhead}
RoBuSt has a redundancy of $O(\log n)$.
\end{corollary}

%%%%%%%%%%%%%%%%%%%%%%%%%%%%%%%%%%%%%%%%%%%%%%%%%%%%%%%%%%%%%%%%%%%%
\section{The Write Protocol}\label{sec:write-protocol}
In the following let $D$ with $|D|\leq (1-\delta)n$ and $\delta <1/72\cdot n^{1/\log\log n}$, be the set of data items for which intact servers received write requests. 
For a data item $d$ that is stored in the system denote the $c$ pieces that have been created from $d$ using Reed Solomon coding as $d_1,\ldots,d_c$.
Furthermore, denote the server that is holding $d_1$ (after the pieces have been spread over the $n$ servers) as the server \mydef{maintaining} $d$.

%%%%%%%%%%%%%%%%%%%%%%%%%%%%%%%%%%%%%%%%%%%%%%%%%%%%%%%%%%%%%%%%%%%%
\subsection{Preprocessing Stage}\label{sec:write-preprocessing}

In this stage, for each crashed server $s_i$, a unique intact server is determined, 
denoted as the \mydef{representative} of $s_i$, such that at the end of this stage each crashed server is the representative of at most two other servers.
The idea of the representatives is to let them take over the roles of the according crashed servers in actions (e.g. routing, computations) the crashed servers were supposed to perform. 
For this, we additionally need to ensure that each intact server knows the representatives of all crashed servers it is connected to in the underlying $k$-ary butterfly.

The determination of the representatives and the introduction of the representatives to the appropriate servers can be done in the same manner as in the \emph{butterfly completion stage} of \cite{iris}, which can be carried out in $(2+o(1))\log{n}$ rounds with a congestion of at most $O(\log{n})$ (see Lemma~2.11 in \cite{iris}).
In contrast to \cite{iris}, we do not need to compute a so-called \emph{decoding depth} here that gives information about the minimum level of the butterfly that the decoding must be initiated from, which would take $O(\log{n})$ rounds.
In the following, we denote by $\rep{i}$ the representative of $s_i$ if $s_i$ is crashed or $s_i$ itself otherwise.

%%%%%%%%%%%%%%%%%%%%%%%%%%%%%%%%%%%%%%%%%%%%%%%%%%%%%%%%%%%%%%%%%%%%
\subsection{Writing Stage Overview}\label{sec:writing-stage-overview}
In order to keep the specification of our system simple, we first give a high-level overview of how a set of write requests is handled.
Further details are given in the following subsection.

The Writing Stage consists of up to $\numblocks + 1$ phases.
Each phase $\curphase \in \{0, \dots, \numblocks\}$ deals with a single bucket $B_\curphase$ from zone $\curphase$ only and receives a set of data items $D_\curphase$ to be inserted into $B_\curphase$.
At the beginning, $D_0 := D$ is the set of all data items for which there are write requests.
In the following, $\datafrom{B_\curphase}$ denotes the set of data items that are stored in bucket $B_\curphase$ (at the beginning of phase $\curphase$).
Phase $\curphase \in\{0,\ldots,\numblocks\}$ consists of the following steps.
\begin{enumerate}
  \item Completely decode $B_\curphase$ and send all decoded pieces of a data item $d\in\datafrom{B_\curphase}$ to the server maintaining $d$ (for details, see the appendix).
  \label{en:write-1}
  \item If $|\datafrom{B_\curphase} \cup D_\curphase| \leq 2n$: Add the data items from $D_\curphase$ to $\datafrom{B_\curphase}$, choose $c$ new hash functions $h_1,\ldots,h_c:U\to V$ uniformly at random for $B_\curphase$, and reencode $B_\curphase$ (see appendix and below).\label{en:write-2}
  \item Else ($|\datafrom{B_\curphase} \cup D_\curphase| > 2n$): 
  \label{en:write-3}
  \begin{enumerate}
    \item The intact servers agree on a subset $D_{\curphase+1}\subseteq\datafrom{B_\curphase}\cup D$ of size $n$ with the property that for all $d,d' \in D_{\curphase+1}$, $\bit{d}{\curphase}=\bit{d'}{\curphase}=b \in \{0,1\}$ (for details, see the appendix).
    \label{en:write-3-a}
    \item Reencode the data items in $(D_\curphase\cup\datafrom{B})\setminus D_{\curphase+1}$ in bucket $B_\curphase$ and\\
    choose $c$ new hash functions $h_1,\ldots,h_c:U\to V$ uniformly at random for $B_\curphase$.
    (see below)
    \label{en:write-3-b}
    \item Set $B_{\curphase+1} := \zerochildarg{B_{\curphase+1}}$ if $b=0$ and $B_{\curphase+1} := \onechildarg{B_{\curphase+1}}$ if $b=1$ and propagate the data items in $D_{\curphase+1}$ to the next phase
    (for details, see the appendix).
    \label{en:write-3-c}
  \end{enumerate}
\end{enumerate}

Each phase of the Writing Stage can be performed in $O(\log{n})$ rounds with a congestion of $O(\log{n})$ at each server in each round (see appendix).
Since there are at most $O(\log{n})$ phases in the Writing Stage, the overall runtime is $O(\log^2{n})$ rounds.

%%%%%%%%%%%%%%%%%%%%%%%%%%%%%%%%%%%%%%%%%%%%%%%%%%%%%%%%%%%%%%%%%%%%

\subsubsection{Encoding of a Bucket}\label{sec:encoding}
In the following we describe how a set of data items is reencoded into a bucket,
as required in step~\ref{en:write-2} and step~\ref{en:write-3-b}.
Note that the reencoding of a bucket does not only consist of the simple encoding of the data items belonging to that bucket but it consists of some additional steps, as described in the following.

First, in contrast to IRIS, $\rep{1}$ chooses $c$ hash functions $h_1, \dots, h_c : U \rightarrow V$ uniformly at random that will be used to map data pieces of this bucket to servers.
While in IRIS the hash functions that map data pieces to servers are never changed, we need to choose new hash functions for a bucket $B$ whenever $B$ is (re)encoded.
The reason for this is that otherwise the adversary would be able to generate write requests that overload certain servers.

Note that the hash functions need to satisfy certain expansion properties, but if $c$ is chosen sufficiently large ($c \geq 18\log{m}$) they do so, \whp (more information is provided in the appendix).
After that, $\rep{1}$ distributes the $c$ hash functions to all other intact servers $\rep{i}$. 
This distribution can be realized by simply broadcasting the hash functions in the $k$-ary butterfly from level $\log_kn$ to level $0$.
In addition, $s_1$ distributes a current timestamp $t(B_\curphase)$ to all other intact servers and each intact server $\rep{i}$ sets its current timestamp for bucket $B_\curphase$ to that value.
Each server $\rep{i}$ now creates for each data item which it maintains or which it has received write requests for and which are not propagated to the next phase the $c$ pieces $d_1,\dots,d_c$ of $d$ using Reed Solomon coding (Section~\ref{sec:underlying-ds}).
Here, $d_j$, $j\in\{1,\ldots,c\}$, is supposed to be sent to the server $s'$ responsible for $h_j(d)$ or to its representative if $s'$ is crashed.
Unfortunately, a server $\rep{i}$ does not necessarily know the representative of the server $s'$ if that server is crashed.
Thus, instead of sending the data pieces directly, the servers initiate a bottom-up routing in the underlying $k$-ary butterfly in order to determine the representative of $h_j(d)$ for each $1 \leq j \leq c$.
Obviously, this takes only $\log_k{n}$ rounds and can be performed with a congestion of $O(k)$ per node.
Once $\rep{i}$ knows the representative of $h_j(d)$, it directly sends $d_j$ to $h_j(d)$ for all $1 \leq j \leq c$.

After the pieces of data items have been distributed, the servers encode the data items in $(\datafrom{B_\curphase} \cup D_\curphase) \setminus D_{\curphase+1}$ in a distributed fashion. 
Note that the set of data blocks for server $i$ in zone $z$ is completely overwritten for each server $\rep{i}$ in this process.
This can be done by a simple top-down approach using the coding strategy for IRIS (see Section~2.1 in~\cite{iris}).
In addition, we also store the timestamp of the bucket along with the data block by appending it to the composed data block.

The following lemma holds during the encoding step, regardless of the current phase. 

\begin{lemma}\label{lem:lookup_maxcrashedservers}
  Assume the adversary blocks less than $(\gamma/2) \cdot 2^{\log_k{n}}$ servers, with $\gamma = 1/36$. Then, for any data item $d$ that is (re-)written during the current period, and any level $0 \leq \ell \leq \log_k{n}$, there are at most $c/6$ pieces of $d$ that are mapped to sub-butterflies $BF(v)$ (for some $v$ at level $\ell$) with at least $\lceil2^{\ell-1}\rceil$ crashed servers in $BF(v)$, \whp
\end{lemma}

\begin{proof}
 In the following, we denote a sub-butterfly $BF(v)$ for some $v \in V$ at level $\ell$ as \emph{blocked at level $\ell$} if at least $\lceil2^{\ell-1}\rceil$ servers in $BF(v)$ are crashed (note that we need the ceiling function only for the special case $\ell = 0$).
 Let $d$ be a data item, and let $0 \leq \ell \leq \log_k{n}$ be a fixed level in the underlying $k$-ary butterfly.
 First of all, we show that the fraction of blocked sub-butterflies at level $\ell$ is at most $\gamma$.
 Using the Chernoff bounds \cite{chernoff}, we can conclude from this that the number of pieces of $d$ that are mapped to blocked sub-butterflies are at most $c/6$ with high probability.
  
 Recall that each sub-butterfly at level $\ell$ contains exactly $k^\ell$ servers.
 Obviously, for $\ell = 0$, the fraction of crashed servers at level $\ell$ is upper bounded by $\gamma / 2 < \gamma$.
 Thus, in the following, we assume $1 \leq \ell \leq \log_k{n}$.
 %In the following, we denote a sub-butterfly $BF(v)$ for some $v \in V$ as \emph{blocked} if more than $2^\ell/2$ servers in $BF(v)$ are blocked.
 Let $b$ be the number of blocked sub-butterflies at level $\ell$.
 Then, there exist at least $b \cdot 2^{\ell-1}$ crashed servers.
 On the other hand, the adversary can block only less than $\gamma/2 \cdot 2^{\log_k{n}}$ servers.
 Hence, $b \cdot 2^{\ell-1} < \gamma \cdot 2^{\log_k{n}-1}$ which is equivalent to $b < \gamma \cdot 2^{\log_k{n}-\ell}$.
 Recall that there are exactly $k^{\log_k{n}-\ell}$ sub-butterflies at level $\ell$.
 This yields that the fraction of blocked sub-butterflies at level $\ell$ is at most $\gamma \cdot \frac{2^{\log_k{n}-\ell}}{k^{log_k{n}-\ell}} \leq \gamma$.
Using the Chernoff bounds it is easy to show that at most $c/6$ pieces of $d$ are mapped to a blocked sub-butterfly, \whp
 
 % Next, for $1 \leq i \leq c$, we define a random variable $X_i \in \{0, 1\}$, where $X_i = 1$ if and only if the $i$-th piece of $d$ is mapped into a blocked sub-butterfly.
 % Further we define the random variable $X := \sum_{i=1}^{c} X_i$, which counts the number of pieces of $d$ that are mapped into a blocked sub-butterfly.
 % From the previous paragraph, we know that $Pr[X_i=1]\leq \gamma$, which implies $E[X] \leq \gamma \cdot c$.
 % Since $X_1, \dots, X_c$ are independent binary random variables, we can apply the Chernoff bounds, which yield $Pr[X \geq c/6] \leq e^{-(\frac{1}{6\gamma}-1)\gamma c/3} \leq e^{-(5/108) \cdot c}$.
 % Due to $c = \Theta(\log{m})$, this can be made arbitrarily small. 

\end{proof}

The lemma plays an important role in the proof of the correctness of the Lookup Protocol.

%%%%%%%%%%%%%%%%%%%%%% LOOKUP %%%%%%%%%%%%%%%%%%%%%%%%%%%%%%%%%%%%%%%%%%%%%%%%%%%%

\section{The Lookup Protocol}\label{sec:lookup-protocol}
In order to keep the specification of our system simple, we provide the description of the lookup protocol as a separate protocol that is executed after the execution of the Write Protocol.
The lookup protocol is divided into two stages: the Preprocessing Stage (Section~\ref{sec:lookup-preprocessing}) and the Zone Examination Stage (Section~\ref{sec:lookup-zone-examination}).
The former is similar to the Preprocessing Stage of the Write Protocol (Section \ref{sec:write-preprocessing}).
The latter is performed for each zone individually and split into two further stages: the Probing Stage and the Decoding Stage.
The basic idea of the Probing Stage is to answer a request by directly collecting a sufficient number of data pieces.
If this is not possible, either because too many of the servers holding a piece are crashed or because of congestion, the Decoding Stage tries to recover a data item by utilizing the distributed coding described in Section \ref{sec:bucket-storage-strategy}.
Note that both a Probing Stage as well as a Decoding Stage can be found in IRIS (\cite{iris}), too.
While they match in their general structure, there are important differences that are caused by the differences in the underlying structure and the implications of the write functionality.
For example, servers may now store obsolete data items without being aware of that.

\subsection{The Preprocessing Stage}\label{sec:lookup-preprocessing}
The Preprocessing Stage is exactly the same as in Section~\ref{sec:write-preprocessing}.
If at least one write request has been handled in the current period, we can thus skip this part and re-use the established $k$-ary butterfly and the unique representatives.

\subsection{The Zone Examination Stage}\label{sec:lookup-zone-examination}

In the following let $\mathcal{D}$ be the set of data items for which a lookup request arrived at an intact server.
The idea of this stage is to successively perform a lookup for each $d\in\mathcal{D}$ in each zone until a copy of $d$ has been found and returned to the appropriate server.
The zone examination stage is performed for at most $\numblocks + 1$ zones starting with zone $0$.

In each phase $z\in\{0,\ldots,\numblocks\}$, beginning with $z=0$, each server with an unserved lookup request for some data item $d$ initiates a lookup request for $d$ in bucket $\fbucket{z}{d}$.
Any server that receives a copy of the data item it requested during the lookup in zone $z$, as described in the following, returns that copy and is finished.
All remaining lookup requests are handled in the next phase, phase $z:= z + 1$.
This procedure is repeated until each lookup request is served.

Handling a set of lookup requests in one phase $z$ is done by performing the Probing Stage and the Decoding Stage as described in the following.

%%%%%%%%%%%%%%%%%%%%%%%%%%%%%%%%%%%%%%%%%%%%%%%%%%%%%%%%%%%%%%%%%%%%%%%%%%%%%%%%
\subsubsection{Probing Stage}\label{sec:lookup-probing}

In the following let $s$ be an intact server that has an unserved lookup request for a data item $d$ at the beginning of phase $z$. 
The idea of the Probing Stage is to either achieve $c/3$ up-to-date pieces such that $d$ can be recovered. Or to assign the request for $d$ to a level $\{1,\ldots,\log_kn\}$ (as defined later) in order to further handle the request in the next stage, the Decoding Stage.
In the following, for a server $s'$, an index $i\in\{1,\ldots,c\}$, and a data piece $d'$ we denote by $P_i(s',d')$ the unique path of length $\log_kn$ in the $k$-ary butterfly from the butterfly node on level $\log_kn$ emulated by $s'$ to the butterfly node on level $0$ emulated by the server that is responsible for $h_i(d')$.

On a high level view, in phase $z$, server $s$ performs the following steps.
\begin{enumerate}
  \item Acquire current hash functions and timestamp \timestamp{d} for bucket $\fbucket{z}{d}$.\label{en:probing-acquire-hash-functions}

  \item Choose $c$ intact servers $s(d_1), \dots, s(d_c)$ uniformly and independently at random.
  \label{en:probing-choose-servers}

  \item Send a \probe{d}{i}{\timestamp{d}} message to $s(d_i)$, $i\in\{1,\ldots,c\}$, in order to initiate the forwarding of the \probenull message along the $c$ paths $P_i(s(d_i),d_i)$.
  \label{en:probing-initiate}
\end{enumerate}
Note that acquiring the hash functions in step~\ref{en:probing-acquire-hash-functions} is necessary since $s$ may have been crashed in the last period in which a write occured in bucket $\fbucket{Z}{d}$ (at which the hash functions were replaced).
% Step 1
Acquiring the current hash functions and the timestamp works as follows:
First of all, $s$ randomly chooses $\kappa := \Theta(\log{n})$ intact servers and asks them for their timestamp in bucket $\fbucket{Z}{d}$.
The intact servers can be found in $O(1)$ communication rounds, \whp, by selecting $\kappa$ random servers in each round until $\kappa$ intact servers have been found.
Let $\timestamp{d}$ be the maximum timestamp $s$ received.
If $\timestamp{d}$ is greater than the timestamp $s$ stores for $\fbucket{Z}{d}$, $s$ knows that it does not have the current hash functions and asks one server from which it received $\timestamp{d}$ for the $c$ hash functions for bucket $\fbucket{Z}{d}$.
Note that during this process each server only receives $O(\log{n})$ requests throughout this process, \whp

% Step 2
Once $s$ knows the correct hash functions, its goal is to retrieve at least $c/3$ pieces of $d$.
Since contacting the servers holding the $c$ pieces of $d$ directly may cause a too high congestion at these servers, we use the method of forwarding $c$ probes from uniformly chosen intact servers $s(d_1),\ldots,s(d_c)$ to the servers responsible for the $c$ pieces of $d$ along the $c$ paths $P_1(s(d_1),d_1),\ldots,P_c(s(d_c),d_c)$ (step~\ref{en:probing-choose-servers}, step~\ref{en:probing-initiate}).
Analogously to step~\ref{en:probing-acquire-hash-functions} choosing the $c$ intact servers in step~\ref{en:probing-choose-servers} takes $O(1)$ communication rounds, \whp

In the following we describe how the nodes from the paths $P_1(s_1,d_1),\ldots,$ $P_c(s_c,d_c)$ react on incoming messages during this phase.
Let $u$ be a butterfly node on level $\ell\in\{0,\ldots,\log_kn\}$ that has received a \probe{d}{i}{\timestamp{d}} message.
In order to reduce redundancy $u$ combines probes for the same piece of $d$ (and thus the same target) and $u$ marks itself as the new origin of the probe (technique of \mydef{splitting and combining} \cite{iris}).
In the following we denote a butterfly node $u$ as \mydef{congested} if it has received more than $\alpha\cdot c$ \probenull messages for different probes, for a sufficiently large constant $\alpha>0$.
Whenever $u$ receives a \probe{d}{i}{\timestamp{d}} message, $u$ performs the following steps.

\begin{enumerate}
  \item If $u$ is congested: 
  \item\mbox{}\quad Stop forwarding the probe and send a \failmsg{d}{i}{\ell} message to the origin of\\
  \mbox{}\quad the probe message.
  \item Else:
  \item\mbox{}\quad If $\ell\neq 0$: Forward \probe{d}{i}{\timestamp{d}} message to the butterfly node on level $\ell-1$\\
  \mbox{}\quad\hspace*{3.5em} on the path $P_i(s(d_i),d_i)$.
  \item\mbox{}\quad If $\ell=0$: (probe has reached its destination)
  \item\mbox{}\quad\quad If $u$'s current version of bucket $\fbucket{Z}{d}$ has timestamp \timestamp{d} and the server\\ \mbox{}\quad\hspace*{0.7em} emulating $u$ is not just a representative of $u$:
  \item\mbox{}\quad\quad\quad If $u$ holds piece $d_i$ of $d$: Send requested piece $d_i$ to the origin of the\\
  \hspace*{13.1em} probe message.
  \item\mbox{}\quad\quad\quad Else: Send \notexistsmsg{d} message to the origin of the probe message.
  \item\mbox{}\quad\quad Else: Send \failmsg{d}{i}{0} to the origin of the probe message.%\\[-1em]
\end{enumerate}

If a butterfly node on level $\ell\in\{0,\ldots,\log_kn-1\}$ receives a data item, a \failmsgnull, or a \notexistsmsg{\cdot} message, it forwards this answer to the origin of the request to which this message was an answer to (along the same path that the request was routed).
A butterfly node on level $\log_kn$ emulated by $s(d_i)$, $i\in\{1,\ldots,c\}$, that received an answer for a probe for data piece $d_i$ simply forwards this answer to the server that initiated the forwarding of that probe.
These answers ensure that after $O(\log_kn)$ rounds the server $s$ 
that received a lookup request for a data item $d$
has received for all initially sent \probenull messages a piece of $d$, or a \notexistsmsg{d} message, or the level at which the probing failed.
Depending on which kinds of answers $s$ has received, it reacts as follows:

\begin{itemize}
  \item If $s$ received at least $c/3$ up-to-date pieces of $d$, $s$ recovers $d$ using Reed Solomon coding and answers the request.

  \item Else if $s$ receives a \notexistsmsg{d} message, $s$ answers that the requested data item does not exist in the system.

  \item Else if $s$ receives more than $2c/3$ \failmsg{d}{i}{0} messages, $s$ declares the request for $d$ to \mydef{belong to level $\ell$}, where $\ell \in \{1,\dots,\log_k{n}\}$ is the smallest level that contains at least $5c/6$ active probes for $d_i$, i.e., probes for $d_i$ that successfully passed the probing at level $\ell$ and all levels $\ell' > \ell$.
\end{itemize}

It is easy to see that the Probing Stage takes at most $O(\log{n})$ communication rounds per phase with at most $O(\log^2{n})$ congestion at every server in each round.
Note that if a data item belongs to a level $\ell$, then at least $5c/6$ of its probes successfully pass level $\ell$ and get deactivated later in the probing (i.e., in a level $\ell' < \ell$).
To this end, each data item can either be retrieved successfully (this is the case if $5c/6 > c/3$ pieces pass level $0$) or belongs to a level $1 \leq \ell < log_k{n}$.

For the proof of the correctness of the protocol, the following lemma plays an important role.
\begin{lemma}\label{lem:probe1}
  If the adversary can only block less than $(\gamma/2) \cdot 2^{\log_k{n}}$ servers, then for every $\ell \in \{1,\dots,\log_k{n}\}$, the number of data items belonging to level $\ell$ is at most $\gamma n/k^{\ell-1}$ with $\gamma = 1/36$.
\end{lemma}
The general idea and structure of the proof of Lemma~\ref{lem:probe1} is based on the proof of Lemma~2.16 in \cite{iris}.
In contrast to \cite{iris}, the only level at which requests are aborted due to crashed nodes is level $0$.
In addition, we also have to take into account that nodes may store outdated information because they were blocked at the round in which new data was written.
Besides this, we have a different definition of when a node belongs to level $\ell$ here (we require at least $5c/6$ active probes instead of $c/2$) and a different value of $\gamma$.

In order to prove Lemma~\ref{lem:probe1} we need to introduce the following definitions:

\begin{definition}[Congested sub-butterfly]
 Let $v$ be a node at level $\ell$ in the butterfly.
 The sub-butterfly $BF(v)$ is called \mydef{congested} at level $\ell$ if the servers in $BF(v)$ receive more than $k^\ell\alpha c/2$ probes for different $d_i$ pieces in total when the requests are processed at level $\ell$.
\end{definition}
\begin{definition}[Congested data item]\label{def:congdataitem}
 A data item $d$ is called \mydef{congested} at level $\ell$ if there exist congested sub-butterflies $BF(\nodeatlvl{i_1}{\ell_1}{d}),\dots,BF(\nodeatlvl{i_r}{\ell_r}{d})$ with $l_i \geq \ell - 1$, $r = c/6$, and $i_1,\dots,i_r$ being pairwise different.
\end{definition}

As a crucial ingredient for the proof of Lemma~\ref{lem:probe1}, we require the hash functions $h_1,\dots,h_c$ to satisfy a certain expansion property, which holds if the hash functions are chosen uniformly and independently at random, \whp.
For this, we need the following definitions.

\begin{definition}[$b$-bundle]
 Given a set $S \subset U$ of keys and a $k \in \mathbb{N}$, we call $F \subseteq S \times \{1,\dots,c\}$ a \mydef{$b$-bundle} of $S$ if every $d \in S$ has exactly $b$ many pairs $(d,i)$ in $F$.
\end{definition}

\begin{definition}[$(b,\sigma)$-expander]
 For any sub-butterfly $B$ let $V(B)$ be the set of servers emulating the nodes of $B$.
 Let $\mathcal{H}$ be a collection of hash functions $h_1, \dots, h_c$.
 Given $h_1, \dots, h_c$ and a level $\ell \in \{0,\dots,\log_k{n}\}$, we define $\serversinvolved{F}{\ell}{S} := \bigcup_{(d,i)\in F}V(BF(\nodeatlvl{i}{\ell}{d}))$.
 Given a $0 < \sigma < 1$, we call $\mathcal{H}$ a \mydef{$(b,\sigma)$-expander} if for any $0 \leq \ell < \log_k{n}$, any $S\subseteq U$ with $|S|\leq \sigma n/k^\ell$, and any $b$-bundle $F$ of $S$, it holds that $|\Gamma_{F,\ell}(S)|\geq k^\ell |S|$.
\end{definition}

The following Claim can be proven analogously to Claim 2.13 of \cite{scheideler2009}.

\begin{claim}\label{claim:bundles}
  If the hash functions $\mathcal{H}=\{h_1,\dots, h_c\}$ are chosen uniformly and independently at random, $m=|U|$ sufficiently large, and $c\geq 18\log m$, then $\mathcal{H}$ is a $(c/6,1/36)$-expander, \whp
\end{claim}

\begin{proofof}{Lemma~\ref{lem:probe1}}
For the proof, we distinguish between level $1$ and all other levels $\ell > 1$.
To simplify, whenever we say that a piece $\datap{d}{i}$ of a data item $d$ is \emph{aborted at level $\ell$}, we mean that the probing for $\datap{d}{i}$ did not successfully pass level $\ell$ (but was answered with a \failmsg{d}{i}{\ell} message).
Recall that there are two reasons for a request for a piece $\datap{d}{i}$ to be aborted: Either due to an excessive congestion (at any node at level $\ell > 0$) or because the server responsible for $\datap{d}{i}$ is crashed or outdated (at level $0$). 

Note that whenever a data item $d$ belongs to level $1$, then more than $3c/6$ of the $c$ pieces of $d$ must have been aborted at level $0$ due to outdated or crashed nodes at level $0$ ($5c/6$ pieces of $d$ successfully passed level $1$ by the definition of when a data item belongs to level $1$ and if at least $c/3$ pieces would have passed level $0$ successfully, the data item would have been answered already).
First of all, Lemma~\ref{lem:lookup_maxcrashedservers} yields that at most $c/6$ of the $c$ pieces of any data item can have been aborted at level $0$ due to outdated nodes.
Thus, for any data item $d$ belonging to level $1$, more than $c/6$ pieces of $d$ must be aborted at level $0$ due to crashed nodes.
We will now bound the maximum number of these data items.
Let $S$ be a maximum set of data items that belong to level $1$.
We will show: $|S| < \gamma n$.
We now construct a set $F$ in the following way: for each $d \in S$, we choose $c/6$ indices $i$ with the property that $\datap{d}{i}$ is aborted at level $0$ due to crashed nodes and add these $(d,i)$ to $F$.
Note that $F$ is a $c/6$-bundle $F$ of $S$.
Since the adversary can block only less than $\gamma n$ servers, the number of servers covered by all $BF(\nodeatlvl{i}{0}{d})$ with $(d,i) \in F$ is less than $\gamma n$.
Since $\serversinvolved{F}{0}{S}$ is exactly the set of these servers, it holds: $|\serversinvolved{F}{0}{S}| < \gamma n$.
On the other hand, we know from Claim~\ref{claim:bundles} that for any $c/6$-bundle $F'$ of $S'$ with $|S'| \leq \gamma n$, $|\serversinvolved{F'}{0}{S'}| \geq |S'|$.
Note that this also implies that for any $c/6$-bundle $F'$ of $S'$ with $|S'| \geq \gamma n$, $|\serversinvolved{F'}{0}{S'}| \geq \gamma n$.
Now, assume for contradiction that $|S| \geq \gamma n$.
This yields $|\serversinvolved{F'}{0}{S'}| \geq \gamma n$.
Since this is a contradiction to what we said before, $|S| < \gamma n$ must hold and thus the number of data items belonging to level $1$ must be bounded by $\gamma n$.

Next, for any level $\ell > 1$, we bound the number of data items belonging to level $\ell$.
First of all, note that the only reason for a piece of data item to be aborted on a level $\ell \geq 1$ is due to congestion at a node at level $\ell$.
Second, note that it can be shown that whenever a $\probe{d}{i}{\timestamp{d}}$ is aborted on level $\ell \geq 1$ due to congestion, then $BF(\nodeatlvl{i}{\ell}{d})$ is congested \whp (see Claim 2.18 of \cite{iris}).
Thus, whenever a data item $d$ is declared to belong to a level $\ell > 1$, then at least $c/6$  $\probe{d}{i}{\timestamp{d}}$ messages have been deactivated at level $\ell - 1$ or higher because of congested sub-butterflies, i.e. $d$ is congested at level $\ell - 1$ (see Def.~\ref{def:congdataitem}).
Thus, if many data items belong to level $\ell$, then many sub-butterflies must be congested at level $\ell -1$.
However, as we will prove, only a constant fraction of the sub-butterflies can be congested at level $\ell -1$, which implies that only a constant fraction of all data items can belong to level $\ell$.

  Fix $1 < \ell \leq \log_k{n}$.
  As mentioned before, we will now bound the number of data items that are congested at level $\ell -1$.
  Let $S$ be a maximum set of data items that are congested at level $\ell -1$.
  We will show: $|S| < \gamma n/k^{\ell-1}$.
  Again, we construct a $c/6$-bundle $F$ of $S$ (adding, for each $d \in S$, $c/6$ indices $i$ to $F$ with the property that $BF(\nodeatlvl{i}{l_i}{d})$ is congested).
  We first show that for $\alpha$ sufficiently large, less than a fraction of $\gamma$ of all butterflies at level $\ell-1$ can be congested.
  Recall that a sub-butterfly on level $\ell-1$ is congested if it receives more than $\alpha c k^{\ell-1}/2$ probes for different $(d,i)$-pairs.
  Let $\delta$ be the maximum fraction of servers the adversary may block.
  Since there are at most $(1-\delta)n$ lookup requests in total, at most $c(1-\delta)n$ probes arrive at level $\ell-1$.
  Thus, at most $c(1-\delta)n/(\alpha c k^{\ell-1}/2) = 2(1-\delta)n/(\alpha k^{\ell-1})$ sub-butterflies can be congested at level $\ell-1$.
  Since there are exactly $n/k^{\ell-1}$ disjoint sub-butterflies at level $\ell-1$, the fraction of congested sub-butterflies at level $\ell-1$ is upper bounded by $2(1-\delta)/\alpha \leq 2/\alpha$.
  Hence, for $\alpha > 2/\gamma$, less than a $\gamma$-fraction of the sub-butterflies on level $\ell-1$ can be congested.
  That is, all of the congested sub-butterflies $BF(\nodeatlvl{i}{l_i}{d})$ with $(d,i) \in F$ together contain less than a $\gamma$-fraction of the sub-butterflies on level $\ell-1$.
  This implies $|\serversinvolved{F}{\ell-1}{S}| < \gamma n$.
  
  On the other hand, from Claim~\ref{claim:bundles}, we can deduce that for any $c/6$-bundle $F$ of $S'$ with $|S'| \geq \gamma n/k^{\ell-1}$, $|\serversinvolved{F}{\ell-1}{S'}| \geq \gamma n$.
  By assuming for contradiction that $|S| \geq \gamma n/k^{\ell-1}$, we can deduce that $|\serversinvolved{F}{\ell-1}{S}| \geq \gamma n$, which is a contradiction in this case, too.
  Thus, $|S| < \gamma n/k^{\ell-1}$.
  
  Therefore, less than $\gamma n / k^{\ell-1}$ are congested at level $\ell-1$.
  For the remaining data items, at least $5c/6$ pieces are not congested.
  Thus, these data items do not belong to level $\ell$.
  This finishes the proof. 
\end{proofof}

\subsubsection{Decoding Stage}\label{sec:lookup-decoding}
The Decoding stage proceeds in $\log_k{n}$ sub-phases. 
In the following, for a server $s$ that holds a lookup request for some data item $d$ that has not been answered before this \subphase, we define $s_i^{(\ell)}(d)$ as the node at level $\ell$ on the unique path of length $\log_k{n}$ from the butterfly node on level $\log_k{n}$ emulated by $s_i(d)$ to the butterfly node on level $0$ responsible for $h_i(d)$.

On a high level view, the Decoding Stage works as follows:
During each \subphase $1 \leq \ell \leq \log_k{n}$, starting with level $1$, we try to recover the data items belonging to level $\ell$.
In order to recover a data item $d$, we need to collect at least $c/3$ pieces of $d$.  
To do so, we randomly choose $5c/6$ requests for pieces of $d$ that were active at level $\ell$ in the Probing Stage and for each of these pieces $d_i$ we determine whether $BF(\nodeatlvl{i}{\ell}{d})$ can be decoded without congestion (as described later).
If $BF(\nodeatlvl{i}{\ell}{d})$ can be decoded without congestion, the decoding is initiated and the result of this is sent back to the origin. (Throughout the whole process, we use the same combining/splitting approach of messages as in the Probing Stage.)
Otherwise, the origin is informed that the according piece of $d$ could not be decoded.
If for a data item $d$ not sufficiently many (i.e., less than $c/3$) pieces could be recovered, the request for $d$ is declared to belong to level $\ell + 1$ and will be considered again in the next \subphase.
Note that requests for non-existing data items may be handled in the Decoding Stage.
However, these can be treated as existing items (with the only difference being that one intact server taking part in the decoding is sufficient to tell that the data item does not exist).

In the following, we describe the operation of any \subphase $\ell$ in more detail.
First of all, each server $s$ that is responsible for a lookup request of a data item $d$ that belongs to level $\ell$ chooses $5c/6$ among the at least $5c/6$ indices of pieces of $d$ that were active at level $\ell$ in the Probing Stage.
For such a piece $d_i$ of $d$ with current timestamp $t$, $s$ sends a $\decodemsg{d}{i}{t}$ message from $\nodeatlvl{i}{\log_k{n}}{d}$ to $v := \nodeatlvl{i}{\ell}{d}$ (which is done by simply routing through the $k$-ary butterfly into the direction of $h_i(d)$ for $\ell$ rounds).
In order to determine whether $BF(v)$ can be decoded without congestion, $v$ first checks whether it is congested, i.e., it received more than $\beta c k$ $\decodemsgnull$ messages for a sufficiently large constant $\beta$ and, if not, then issues a $\decodecheckmsg{d}{i}$ message, which is spread to all nodes in $UT(v)$. 
During this spreading, whenever a further forwarding of all messages received by a node $u$ at a level $\ell-\kappa$, $1\leq \kappa < \ell$, could lead to congestion (i.e., $u$ received more than $\beta c k$ $\decodecheckmsg{d'}{i'}$ messages for distinct $(d',i')$ pairs), $u$ stops the forwarding of all messages and instead spreads a $\congmsg{\cdot}$ message in $BF(u)$.
In addition, it sends a $\failmsgnull$ message to all neighbors at level $\ell-\kappa+1$.
Each node on a level $\ell'$, $\ell-\kappa+1 \leq \ell' < \ell$, that receives such a \failmsgnull message forwards this message to all neighbors at level $\ell'+1$ from which it received a \decodecheckmsgnull mesage.
By this it is ensured that whenever a node in $BF(u)$ is congested each node $v'$ at level $\ell$ with $v'\in BF(u)$ receives a \failmsgnull message after at most $2\ell$ rounds.
Each node $u'$ at level $\ell-\kappa$, $1\leq \kappa < \ell$, that received a $\congmsg{}$ message initiates the same spreading of \congmsg{} messages in $UT(u')$.
If $v$ had not been congested before the spreading and $v$ has not received any $\failmsgnull$ message after $2\ell$ rounds, it knows that any piece of a data item for which $v$ received a $\decodemsgnull$ message can be decoded if not outdated nodes in $BF(v)$ forbid this.
Thus, it initiates the decoding for each of the pieces, which may fail due to outdated nodes.
If the decoding is possible, it recovers all of these pieces within $O(\ell)$ communication rounds with a congestion of at most $\beta c k^2$ per node (using the distributed decoding described in \cite{iris}).
These are then forwarded to the origins of the requests.
If, however the decoding fails, or if $v$ was congested or received a $\failmsgnull$ message, it sends a \failmsgnull message to the origins of the \decodemsgnull messages it received (which, again, are forwarded up to the initiator of that \decodemsgnull message).
Finally, if a server $s$ that is responsible for a lookup request of a data item $d$ receives at least $c/3$ successfully decoded pieces, it determines $d$ and answers the request.
Otherwise, it changes the request to belong to level $\ell + 1$ such that it will be processed again in the next \subphase.

It is easy to see that the Decoding Stage satisfies the following property:
\begin{lemma}
The Decoding Stage takes at most $O(\log{n})$ communication rounds per \subphase with at most
$O(\log^3 n)$ congestion in every node at each round, \whp
\end{lemma}

Similarly to Lemma~\ref{lem:probe1} of the Probing Stage, for the Decoding Stage the following lemma holds:

\begin{lemma}\label{lemma:decoding}
At the beginning of each \subphase $\ell \in \{1, \dots, \log_k{n}\}$, the number of data items with requests belonging to level
$\ell$ is at most $\varphi n/k^{\ell}$ with $\varphi =\Theta(k)$.
\end{lemma}

For the proof of Lemma~\ref{lemma:decoding} we need the following definitions:

\begin{definition}[Blocked sub-butterfly]\label{def:crashed-butterfly}
 Let $v$ be a node at level $\ell$ in the butterfly.
 The sub-butterfly $BF(v)$ is called \mydef{blocked} at level $\ell$ if at least $2^{\ell-1}$ servers from $BF(v)$ are crashed. 
\end{definition}
\begin{definition}[Congested sub-butterfly]
 Let $v$ be a node at level $\ell$ in the butterfly.
 The sub-butterfly $BF(v)$ is called \mydef{congested} at level $\ell$ if the servers in $BF(v)$ receive more than $\beta c k$ requests for different $d_i$ pieces in total when the requests are processed at level $\ell$.
\end{definition}
\begin{definition}[Blocked/Congested data item]\label{def:blocked_congested_data_item}
 A data item $d$ is called \\
 \mydef{blocked}/\mydef{congested} at level $\ell$ if there exist blocked/congested sub-butterflies \\
 $BF(\nodeatlvl{i_1}{\ell_1}{d}),\dots,BF(\nodeatlvl{i_r}{\ell_r}{d})$ with $l_i \geq \ell$, $r = c/6$, and $i_1,\dots,i_r$ being pairwise different.
\end{definition}

Furthermore, we need the following claims.
\begin{claim}\label{claim:outdated}
  For any data item $d$ which is neither blocked nor congested, at most $c/6$ pieces of $d$ can fail due to outdated servers.
\end{claim}
\begin{proof}
 Assume a data item $d$ is neither congested nor blocked.
 This means that less than $c/6$ pieces of $d$ are congested and less than $c/6$ pieces of $d$ are blocked.
 The latter implies that the adversary blocks less than $2^{\ell-1}$ servers from the sub-BFs $BF(\nodeatlvl{i_1}{l_1}{d}), BF(\nodeatlvl{i_1}{l_1}{d}), \dots$ for the remaining pieces $\datap{d}{i_1}, \datap{d}{i_2}, \dots$ of $d$.
 By Lemma~\ref{lem:lookup_maxcrashedservers}, for all but $c/6$ of these pieces, less than $2^{\ell-1}$ of the servers in $BF(\nodeatlvl{i_1}{l_1}{d}), BF(\nodeatlvl{i_1}{l_1}{d}), \dots$ can be outdated regarding $d$.
 Thus, for all but $c/6$ of these pieces, less than $2^{\ell-1}+2^{\ell-1} = 2^{\ell}$ servers can be crashed or outdated, which, by Claim~2.17 of \cite{iris} means that these pieces can be recovered.
 Thus, at most $c/6$ pieces of the data items that are neither blocked nor congested can fail due to outdated servers.
 \end{proof}

Now we are ready to prove Lemma~\ref{lemma:decoding}. The proof is similar to the proof of Lemma~2.21 of \cite{iris} with the main difference being that we additionally need to handle outdated data items here.

\begin{proofof}{Lemma~\ref{lemma:decoding}}
  In the following, let $\gamma = 1/36$ and $\varphi = 3\gamma k$.
 We prove the lemma by induction on $\ell$.
 The basis ($\ell=1$) holds by Lemma~\ref{lem:probe1}.
 For the induction step, let $\ell \in \{1,\dots,\log_k{n} - 1\}$ and assume that the induction hypothesis holds for level $\ell$.
 We show that the number of data items that will be propagated to level $\ell+1$ during \subphase $\ell$ is at most $2\gamma n/k^{\ell}$.
 Together with Lemma~\ref{lem:probe1}, this means that at the beginning of \subphase $\ell+1$, at most $\gamma n/k^{\ell} + 2\gamma n/k^{\ell}$ data items belong to level $\ell+1$, which is equal to $\varphi n/k^{\ell+1}$ and thus proves the induction step.
 
 Recall that in \subphase $\ell$, requests for $5c/6$ pieces of each data item belonging to level $\ell$ are sent.
 Note that any request for a piece $\datap{d}{i}$ of a data item $d$ in \subphase $\ell$ of the Decoding Stage can only be aborted for one of the following three reasons: First, that too many servers storing information about $\datap{d}{i}$ are crashed in the current period.
 Second, that too many servers storing information about $\datap{d}{i}$ are outdated (i.e., they were crashed when the bucket storing $d$ was last updated).
 Third, due to congestion in \subphase $\ell$ of the Decoding Stage.
 However, it can be shown that if at least $c/6$ requests for a data item $d$ are aborted during the decoding in \subphase $\ell$ due to too many crashed nodes, then $d$ is blocked at level $\ell$ and if at least $c/6$ requests for a data item $d$ are aborted during the decoding in \subphase $\ell$ due to congestion, then $d$ is congested at level $\ell$ \whp
 The former claim is an implication of Claim~2.17 of \cite{iris}, and the latter follows by definition and the algorithm performed in the decoding stage.
Claim~\ref{claim:outdated} now implies that for the data items which are neither blocked nor congested, at least $5c/6 - c/6 - c/6 - c/6 = c/3$ pieces can be recovered correctly, which means that they can be answered after \subphase $\ell$.
Thus, in the following, we will show that at most $\gamma n/k^{\ell}$ data items are blocked at level $\ell$ and that at most $\gamma n/k^{\ell}$ data items are congested at level $\ell$.

  %Blocked-Teil
  First of all, we prove that the number of blocked data items in \subphase $\ell$ is upper bounded by $\gamma n/k^\ell$.
  Let $S$ be a maximum set of data items that are blocked at level $\ell$.
  We will show: $|S| < \gamma n/k^\ell$.
  Recall that a data item $d$ is blocked at level $\ell$ if there exist at least $r = c/6$ sub-butterflies $BF(\nodeatlvl{i_1}{\ell_1}{d}),\dots,BF(\nodeatlvl{i_r}{\ell_r}{d})$ with $\ell_i \geq \ell$, and $i_1,\dots,i_r$ being pairwise different that are blocked, i.e., each of them contains at least $2^{\ell_i-1}$ crashed servers.
  For each $d \in S$, let $d_{i_1},\dots,d_{i_r}$ be $c/6$ such indices fulfilling this property.
  Further, let $(d,{d_{i_1}}), \dots, (d,{d_{i_r}}) \in F$ for all $d \in S$.
  Then, $F$ is a $c/6$-bundle of $S$.
  Since a sub-butterfly of level $\ell'$ contains $k^{\ell'}$ servers in total, and since a blocked sub-butterfly of level $\ell'$ contains at least $2^{\ell'-1}$ crashed nodes, a $2^{\ell'-1}/k^{\ell'}$ fraction of the servers of a blocked sub-butterfly of level $\ell'$ are crashed, which is at least $2^{\log_k{n}-1}/n$ for any $1 \leq \ell' \leq \log_k{n}$.
  Therefore, if the adversary can only block less than $(\gamma/2)\cdot 2^{\log_k{n}}$ servers, then the number of servers covered by all $BF(\nodeatlvl{i}{\ell_i}{d})$ with $(d,i) \in F$ must be less than $\gamma n$.
  Since $\serversinvolved{F}{\ell}{S}$ is exactly the set of these servers, it holds: $|\serversinvolved{F}{\ell}{S}| < \gamma n$.
  
  On the other hand, we know from Claim~\ref{claim:bundles} that for any $c/6$-bundle $F'$ of $S'$ with $|S'| \leq (1/36) n/k^\ell$, $|\serversinvolved{F'}{\ell}{S'}| \geq |S'| k^\ell$.
  %Clearly, this implies that for any $c/6$-bundle $F$ of $S'$ with $|S'| \geq \sigma n/k^\ell$, $|\serversinvolved{F}{\ell}{S'}| \geq \sigma n$.
  Since $\gamma = 1/36$, this implies that for any $c/6$-bundle $F'$ of $S'$ with $|S'| \geq \gamma n/k^\ell$, $|\serversinvolved{F'}{\ell}{S'}| \geq \gamma n$.
  Now, assume for contradiction that $|S| \geq \gamma n/k^\ell$.
  This yields $|\serversinvolved{F}{\ell}{S}| \geq \gamma n$, which is a contradiction to what we said before.
  Hence, the number of blocked data items at level $\ell$ is less than $\gamma n/k^\ell$.
   
 %Congestion-Teil
  For the upper bound on the number of congested data items, recall that we denote a sub-butterfly $BF(v)$ of a node $v$ as congested if the servers in $BF(v)$ receive more than $\beta c k$ decode messages for different $(d,i)$-pairs.
  For $\beta := 3$, it holds that $\beta c k > 5 \varphi c / (6\gamma)$, which implies that a congested sub-butterfly $BF(v)$ of a node $v$ receives more than $5 \varphi c / (6\gamma)$ $\decodemsg{d}{i}{t}$ messages for different $d$ and $i$.
  By the induction hypothesis and due to the fact that we send $5c/6$ $\decodemsgnull$ messages per data item, there are at most $5c/6 \cdot \varphi n / k^\ell$ $\decodemsgnull$ messages in total, which means that there are less than $\varphi n/k^\ell \cdot 5c/6 \cdot 6\gamma/(5c\varphi) = \gamma n/k^\ell$ congested sub-butterflies of dimension $\ell$.
  
  Let $S$ be a set of data items congested at level $\ell$.
  Similar to the previous part about blocked data items, we can construct a $c/6$-bundle $F$ for $S$.
  Since there are less than $\gamma n/k^\ell$ congested sub-butterflies of dimension $\ell$ and since each sub-butterfly of dimension $\ell$ contains $k^\ell$ nodes, $|\serversinvolved{F}{\ell}{S}| < \gamma n$. 
  On the other hand, if we assume $|S| \geq \gamma n/k^\ell$, Claim~\ref{claim:bundles} yields $|\serversinvolved{F}{\ell}{S}| \geq \gamma n$.
  Since this is a contradiction, we have that the nummber of data items congested at level $\ell$ is less than $\gamma n/k^\ell$.
  
  As stated at the beginning of the proof, this is sufficient to prove the induction step and thus completes the proof of the lemma.
\end{proofof}

The previous lemmas and results imply Corollary~\ref{corollary:main}, which proves Theorem~\ref{theorem:main}.

\begin{corollary}\label{corollary:main}
RoBuSt correctly serves any set of lookup and write requests (with one request per intact server) in at most $O(\log^4n)$ communications rounds, with a congestion of at most $O(\log^3n)$ at every server in each round and a redundancy of $O(\log n)$ if less than $1/72\cdot n^{1/\log\log n}$ servers are crashed.
\end{corollary}

\section{Conclusion and Future Work}
We presented the first scalable distributed storage system that is provably robust against batch-based crash failures with up to $\gamma n^{1/\log\log n}$ crashes allowed ($\gamma>0$ constant).
An interesting question that has not been investigated in this work is whether the techniques that enabled the Enhanced IRIS system \cite{iris} to tolerate a larger number of failed servers could be adapted for RoBuSt in order to increase the number of crashed servers allowed up to $\varrho n$ (for some constant $\varrho>0$) while (as a minor drawback) also increasing the redundancy to $O(\log n)$, such as it is the case in Enhanced IRIS.

Moreover, while we assume batch-based failures, it would be interesting to see whether a scalable distributed storage system can be designed that can tolerate failures occuring at arbitrary points in time.
Dealing with a similar issue, it would also be interesting to enhance our system to allow dynamics (i.e. joins and leaves of servers) in our system in order to model P2P networks.

A further interesting challenge is to enhance our distributed storage system such that additional types of attacks can be handled, for example Byzantine attacks.

\bibliographystyle{plain}
\bibliography{RoBuSt}

\section*{Appendix}

\subsection{Details on Decoding}\label{sec:decoding}
At the beginning of each phase of the Writing Stage, for each data item $d$ belonging to the current bucket $B_\curphase$, all pieces of $d$ are decoded and sent to the server maintaining $d$.
This can be done by a bottom-up approach that proceeds in $\log_k{n} + 1$ rounds:
First of all, note that for each bucket, each server stores the timestamp of when it last took part in a Writing Stage for this bucket.
These timestamps enable a server to identify that it stores outdated information about a bucket.

In round $r \in \{0, \dots, \log_k{n} - 1\}$, each node $u$ at level $\log_k{n} - r$ in the $k$-ary butterfly forwards all decoding information about a node $v$ at level $\log_k{n} - r - 1$ to $v$.
In any round $r' \in \{1, \dots, \log_k{n}\}$ any crashed node $v$ at level $\log_k{n} - r'$ filters out all messages whose timestamp is not a highest among those received and, if the number of remaining messages is at least $k-1$, it can use these messages to decode enough information to function as an intact node from now on.
Note that if $v$ received a timestamp higher than its own one, it behaves like a crashed one from now on.

Any node $w$ at level $0$ that is still crashed after round $\log_k{n} - 1$ can restore the pieces of data items in $w$ by the messages it receives from the nodes at level $1$.
This is due to the following: First of all, the adversary can block less than $\frac{1}{2}\cdot 2^{\log_k{n}}$ servers in the current period only. Secondly, less than $\frac{1}{2}\cdot 2^{\log_k{n}}$ servers can store outdated information about bucket $B_\curphase$, for the same reason.
Thus, less than $2^{\log_k{n}}$ servers can have none or outdated pieces of data items.
Then, from Claim~2.17 of \cite{iris} it follows that all pieces can be recovered at level $0$.

Note that the $c$ hash functions for the pieces of a data items were been chosen uniformly and independently at random when $B_\curphase$ was encoded, and after the adversary had decided on the set of blocked servers.
Thus, each server holds $O(\log{n})$ pieces for bucket $B_\curphase$, \whp
Furthermore, each server maintains at most one data item, \whp
This implies that the above process yields a congestion of at most $O(\log{n})$ \whp

All in all, we have:
\begin{lemma}
 After $\log_k{n}+1$ rounds with a congestion of $O(\log n)$ at each server in each round \whp, each server $\rep{i}$ maintaining a data item $d$ in $B_\curphase$ completely knows $d$, \whp
\end{lemma}

\subsection{Details on Counting and Selection}\label{sec:counting}

In the following we describe the process of determining the number of data items in $\datafrom{B_\curphase} \cup D_\curphase$ and the elements of the set $D_{\curphase+1}$ (if necessary) for a phase $\curphase$ in a distributed fashion in more detail.
For the set $D_\curphase$ of data items to be inserted into $B_\curphase$, we denote by $\requestServer{\curphase}{i} \subseteq D_\curphase$ the set of data items with a write request at server $\rep{i}$. 
Furthermore, we denote by $\bucketServer{\curphase}{i} \subseteq \datafrom{B_\curphase}$ the set of data items from bucket $B_\curphase$ that server $\rep{i}$ maintains.

First of all, each server $\rep{i}$, $i \in \{1,\dots,n\}$ initializes a tuple $(\num{0}, \num{1})$ where $\num{j}$, $j\in\{0,1\}$, is the number of data items $d \in \bucketServer{\curphase}{i} \cup \requestServer{\curphase}{i}$ with $\bit{d}{\curphase+1}=j$.
These tuples are now forwarded bottom up in the underlying $k$-ary butterfly, where each intermediate node sums up all tuples it received and forwards the result to the next smaller level.
More precisely, each server $\rep{i}$ first sends its tuple $(\num{0},\num{1})$ to each of the $k$ neighbors of the node $(\log_k{n},i)$ in the underlying $k$-ary butterfly.
Any intermediate node $v$ on level $\ell$, $0 < \ell < \log_k{n}$, 
sets $\num{j}$, $j\in\{0,1\}$, as the sum of all $k$ received $\num{j}$-values
and sends the tuple $(\num{0},\num{1})$ to its $k$ neighbors on level $\ell-1$ in the underlying $k$-ary butterfly.
Finally, a server $\rep{i}$ on level $0$, sums up all tuples received from neighbors on level $1$ and stores the result in $(\num{0}',\num{1}')$.
The following lemma is easy to check.
\begin{lemma}\label{lemma-counting}
  After $\log_k{n}$ rounds, each server $\rep{i}$ knows the number of data items $d \in \datafrom{B_\curphase} \cup D_\curphase$ with $\bit{d}{\curphase+1}=j$ for all $j \in \{0,1\}$.
  Additionally, in every round, each server $s$ sends and receives at most $2k$ messages.
\end{lemma}
The servers can now compute $\textsf{size}(B_\curphase) := \num{0}'+\num{1}'$ and check whether $\textsf{size}(B_\curphase) > 2n$.
If this is not the case, the current bucket is reencoded together with the items from $D_r$ (see Section~\ref{sec:encoding}) and the Writing Stage finished. 
Otherwise the servers need to degree on the bucket and a set $D_{\curphase+1}$ of $n$ data items from $\datafrom{B_\curphase}\cup D_\curphase$ to be handled in the next phase.
Whether this bucket is either $\zerochildarg{B_\curphase}$ or $\onechildarg{B_\curphase}$ depends on the number of data items with the same $(\curphase+1)$-st bit in $D_\curphase\cup\datafrom{B_\curphase}$.
I.e., if $\num{0}' > n$ the next bucket is $\bucket{\curphase+1}:=\zerochildarg{B_\curphase}$ and we set $j:= 0$, otherwise the next bucket is $\bucket{\curphase+1}:=\onechildarg{B_\curphase}$ and we set $j:=1$.
Then, since $size(B_\curphase) > 2n$, $\num{j}' > n$ must hold.

In the following, the servers determine the set $D_{\curphase+1}$ of $n$ data items that will be propagated to bucket $\bucket{\curphase+1}$, as required in step~\ref{en:write-3-c}.
This is done by a top-down approach in the tree $LT((0,0))$ of the $k$-ary butterfly.
In the following, we assume that each node $v$ in $LT((0,0))$ during the first part (the bottom-up counting) stored the tuples $(t_{1,0},t_{1,1}),(t_{2,0},t_{2,1}), \dots, (t_{k,0},t_{k,1})$ it received from its children $v_1, \dots, v_k$ in $LT((0,0))$ and is now still able to determine 
the value of $t_{i,j}$, $i\in\{1,\ldots,k\}$.
Furthermore, the nodes exchange two different types of messages during this step: $\msgfull$ and $\msgpartly{x}$, $x \in \mathbb{N}$.
At the beginning, $(0,0)$ issues $\msgpartly{size($B_\curphase$)}$ on itself.
Depending on the message a node $v$ receives, it performs the actions described in the following.
\begin{description}
 \item [\textrm{\msgpartly{x}}:] If $v$ is not on level $\log_k{n}$, let $v_1, \dots, v_k$ denote the children of $v$ in $LT((0,0))$. 
        Determine the greatest index $b$ such that $y \leq x$ with 
        $y := \sum_{i=1}^{b} t_{i,j}$.
        Send $\msgfull$ to $v_1, \dots, v_b$.
        If $x-y > 0$, send $\msgpartly{x-y}$ to $v_{b+1}$.
        If $v$ is on level $\log_k{n}$, the server emulating $v$ randomly chooses $x$ data items $d  \in \bucketServer{\curphase}{i} \cup \requestServer{\curphase}{i}$ with $\bit{d}{\curphase+1}=j$.
        These data items belong to $D_{\curphase+1}$ and will be handled in the next phase as if the server emulating $v$ has a new write request for them.

  \item [\textrm\msgfull:] If $v$ is not on level $\log_k{n}$, $v$ sends a $\msgfull$-message to each of its children in $LT((0,0))$.
        If $v$ is on level $\log_k{n}$, the server emulating $v$ removes all data items $d \in \bucketServer{\curphase}{i} \cup \requestServer{\curphase}{i}$ with $\bit{d}{\curphase+1}=j$.
        These data items belong to $D_{\curphase+1}$ and will be handled in the next phase as if the server emulating $v$ has a new write request for them.
\end{description}

The following lemma is easy to check:
\begin{lemma}\label{lemma-determine-bucket}
  After $\log_k{n}$ additional rounds it holds:
  \begin{enumerate}
    \item Each server $\rep{i}$ knows which of the data items in $\bucketServer{\curphase}{i} \cup \requestServer{\curphase}{i}$ are supposed to be encoded in bucket $B_\curphase$ again and which of them are propagated to the next phase.
    \item In every round, each server $s$ sends and receives at most $2k$ messages.
    \item The number of data items that are decided to belong to bucket $B_\curphase$ (and thus will be encoded in this bucket) is at most $2n$.
  \end{enumerate}
\end{lemma}

It remains to distribute for each data item from $D_{\curphase+1}$ a write request among the $n$ servers such that each server $\rep{i}$ is responsible for exactly one of these write requests.
This distribution can easily be achieved by using standard techniques for load balancing in the butterfly in $O(\log{n})$ rounds and a congestion of $O(\log{n})$ at each server.

%%%%%%%%%%%%%%%%%%
\subsection{Figures \& Glossary}

Figure~\ref{fig:probing-stage} and Figure~\ref{fig:decoding-stage} visualize the Probing and Decoding Stage.

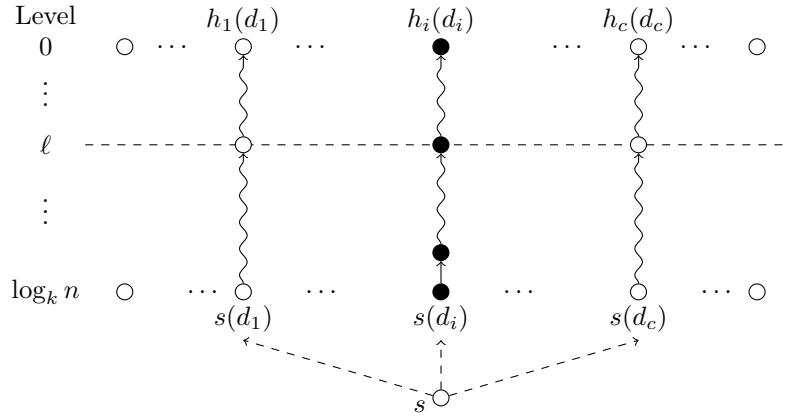
\begin{figure}[htp]
	\centering
	%\documentclass{article}
%\usepackage{ifthen}
%\usepackage{mathptmx}
%\usepackage{tikz}
%\usepackage{color}
%\usepackage[active,pdftex,tightpage]{preview}
%\usepackage{rotating}
%\PreviewEnvironment[]{tikzpicture}
%\PreviewEnvironment[]{pgfpicture}
%\DeclareSymbolFont{symbolsb}{OMS}{cmsy}{m}{n}
%\SetSymbolFont{symbolsb}{bold}{OMS}{cmsy}{b}{n}
%\DeclareSymbolFontAlphabet{\mathcal}{symbolsb}
%
%% distance between two points in a grid
%\newcommand{\gridsize}{11pt}
%
%% cricle size
%\newcommand{\nodesize}{8pt}
%
%\usetikzlibrary{arrows}
%\usetikzlibrary{fit}
%\usetikzlibrary{shadows}
%\usetikzlibrary{patterns}
%\usetikzlibrary{matrix,shapes,calc,snakes}
%\usetikzlibrary{positioning} 
%\usetikzlibrary{backgrounds}
%\usetikzlibrary{calc}
%\usetikzlibrary{decorations.markings}
%
%% define tikz styles used in our figures
%\tikzstyle{every picture}+=[
%node/.style={circle, minimum size=\nodesize, draw=black, inner sep=0,fill=white},       % print a circle
%blacknode/.style={circle, minimum size=\nodesize, draw=black, inner sep=0,fill=black},
%empty/.style={draw=none,fill=none},
%curvedLine/.style={decorate,decoration={snake,amplitude=.5mm}}]
%
%\newcommand{\labelednode}[4]{
%        \node[node] (#2) at (#1){};
%        \node[empty,yshift=1em] at (#1) {#3};
%}
%\newcommand{\xshift}{2}
%\newcommand{\sidelabelednode}[4]{
%        \node[node] (#2) at (#1){};
%        \node[empty,xshift=\xshift em] at (#1) {#3};
%}
%
%\newcommand{\sidedownlabelednode}[4]{
%        \node[node] (#2) at (#1){};
%        \node[empty,yshift=-0.7em,xshift=\xshift em] at (#1) {#3};
%}
%
%
%
%\begin{document}

\begin{tikzpicture}[scale=1.3]
	\def\nodexdist{0.4}
	\def\nodeydist{0.5}
	\def\levell{-2*\nodeydist}
	\def\level0{\nodeydist-0.05em}
	\def\leveld{-5*\nodeydist}
	\def\hidx{8*\nodexdist}
	\def\hcdx{13*\nodexdist}
	\def\hc1x{3*\nodexdist}
	\def\fillcolor{white}
	
	%%%%%%%%%%%%%%%%%%%%%%%%%%%%%%%%%%%%%
	%  Level 0 and 1
	%%%%%%%%%%%%%%%%%%%%%%%%%%%%%%%%%%%%%

	% Level 0 nodes
	\labelednode{0,\level0}{n}{};
	\node[empty] at (\hc1x-0.7,\level0){$\dots$};
	\labelednode{\hc1x,\level0}{h1}{$h_1(d_1)$};
	\node[empty] at (\hc1x+0.7,\level0){$\dots$};
	
	\node[node,fill=black] (hi) at (8*\nodexdist,\level0){};
	\node[empty,yshift=1em] at (hi) {$h_i(d_i)$};
	
	\node[empty] at (\hcdx-0.7,\level0){$\dots$};
	\labelednode{\hcdx,\level0}{hc}{$h_c(d_c)$};
	\node[empty] at (\hcdx+0.6,\level0){$\dots$};
	\labelednode{16*\nodexdist,\level0}{}{};

	% level l line
	\draw[dashed] (-0.4,\levell) -- (6.7,\levell);
	
	% level l nodes
	\node[node] (ul1) at (3*\nodexdist,\levell){};
	
	\node[node,fill=black] (uli) at (8*\nodexdist,\levell){};
	%\node[empty,yshift=0.7em,xshift=1em] at (uli) {$v$};
	
	\node[node] (ulc) at (13*\nodexdist,\levell){};
	
	% level l+1 nodes
	%\node[node,fill=black] (v2) at (8*\nodexdist,\levell-\nodeydist) {};
	%\node[node,fill=black] (v1) at (6*\nodexdist,\levell-\nodeydist) {};
	%\node[node,fill=black] (v3) at (10*\nodexdist,\levell-\nodeydist) {};
	
	% connections from level l+1 to level l
	%\draw[->] (v2) -> (uli);
	%\draw[->] (v1) -> (uli);
	%\draw[->] (v3) -> (uli);

	% Level d nodes
	\foreach \x in {0,3,8,13,16} {
		\ifthenelse{\x=8}{\def\fillcolor{black}}{}
		\node[node,fill=\fillcolor] (u\x) at (\x*\nodexdist,\leveld){};
		\ifthenelse{\x = 16}{}{
			\node[empty,xshift=3em] at (u\x) {$\dots$};
		}
	}
	
	% last level node labels
	\node[empty,yshift=-1em] at (u3) {$s(d_1)$};
	\node[empty,yshift=-1em] at (u8) {$s(d_i)$};
	\node[empty,yshift=-1em] at (u13) {$s(d_c)$};
	
	% server s
	\node[node,yshift=-4em] (server) at (u8) {};
	\node[empty,xshift=-0.8em,yshift=-0.3em] at (server) {$s$};
	
	% connections from s to last level nodes
	\draw[dashed,->] (server) -- ($(u3)+(0,-1.4em)$);
	\draw[dashed,->] (server) -- ($(u8)+(0,-1.4em)$);
	\draw[dashed,->] (server) -- ($(u13)+(0,-1.4em)$);
	
	% u_i path nodes
	\renewcommand{\xshift}{3.1}
	\node[node,fill=black] (ui1) at (8*\nodexdist,-4.2*\nodeydist){};
		
	% connections from level l nodes to level 0 nodes
	\draw[curvedLine,->] (u3) -> (ul1);
	\draw[curvedLine,->] (ul1) -> (h1);
	\draw[curvedLine,->] (u13) -> (ulc);
	\draw[curvedLine,->] (ulc) -> (hc);
	\draw[->] (u8) -- (ui1);
	\draw[curvedLine,->] (ui1) -> (uli);
	\draw[curvedLine,->] (uli) -> (hi);
	
	% Level names
	\node[empty,yshift=1.2em] at (-0.8,\level0) {Level};
	\node[empty] (level0label) at (-0.8, \level0) {0};
	\node[empty] (levelllabel) at (-0.8, \levell) {$\ell$};
	\node[empty] (leveldlabel) at (-0.8, \leveld) {$\log_kn$};	
	\node[empty] at ($(level0label)!0.4!(levelllabel)$) {$\vdots$};
	\node[empty] at ($(levelllabel)!0.4!(leveldlabel)$) {$\vdots$};	
	
\end{tikzpicture}

%\end{document}
	\caption{
		Visualization of the Probing Stage. The curved paths denote the paths $P_i(s(d_i),d_i)$.}
    \label{fig:probing-stage}
\end{figure}

\begin{figure}[htp]
\centering
\begin{tikzpicture}[scale=1.45]
	\def\nodexdist{0.3}
	\def\nodeydist{0.5}

	%  Level 0 nodes
	\foreach \x in {0,2,3,6,9,10,12}
		\node[node] (\x-0) at (\x*\nodexdist,0){};
		
	% Level label
	\node[empty] at (-0.6,0.4){Level};
		
	% h_i(d) node label
	\draw(6-0) node[label=above:$h_i(d_i)$]{};
		
	% Level 0 dots
	\foreach \x in {1,4,8,11}
		\node[empty] (\x-0) at (\x*\nodexdist,0) {$\dots$};
				
	% Last  level nodes
	\def\lastlevel{-4.5*\nodeydist}
	\foreach \x in {0,6,12}
		\node[node] (\x-last) at (\x*\nodexdist,\lastlevel){};
		
	% Last level dots
	\node[empty] at (0.8,\lastlevel) {$\dots$};
	\node[empty] at (2.7,\lastlevel) {$\dots$};
		
	% y-distance to l level
	\def\midlevel{\lastlevel/2-0.8*\nodeydist}

	% level r line
	\draw[dashed] (-0.2,\midlevel) -- (3.8,\midlevel);

	% intermediate u node
	\node[node] (u-mid) at (6*\nodexdist,\midlevel){};
	
	% intermediate u node label
	\node[empty,xshift=2.8em,yshift=-0.8em] at (u-mid){$u=s_i^{(\ell)}(d_i)$};
	
	% parents of itermediate node
	\def\dist{\midlevel+0.3}
	\foreach \x in {0,1,2} {
		\node[node] (u-mid-\x) at (5*\nodexdist+\x*0.2+0.1,\dist){};
		\draw (u-mid-\x) -- (u-mid);
	}
	
	% level r-i line
	\draw[dashed] (-0.2,\midlevel+2*\nodeydist) -- (3.8,\midlevel+2*\nodeydist);
	
	% node v at level r-i
	\node[node] (node-v) at (6*\nodexdist,\midlevel+2*\nodeydist){};
	\node[empty,xshift=0.7em,yshift=0.7em] at (node-v) {$v$};

	% line from last level node to u
	\draw[curvedLine,->] (6-last) -- (u-mid);		
	
	% tree
	\draw (u-mid-0.west) -- (3-0.west);
	\draw (u-mid-2.east) -- (9-0.east);	
	
	% Level labels
	\node[empty] at (-0.5,0) {$0$};
	\node[empty] at (-0.6,\midlevel+2*\nodeydist) {$\ell-\kappa$};
	\node[empty] at (-0.5,\midlevel) {$\ell$};
	\node[empty] at (-0.7,\lastlevel) {$\log_kn$};
	
	% u node label on last level
	\draw(6-last) node[label=below:$s(d_i)$]{};
	
	% T(u,d) label
	\node[empty] at (6*\nodexdist,0.4\dist) {$UT(u)$};
	
% 	% butterfly box
% 	\node[fit={
% 		($(0-0)+(-2.5em,1.6em)$) 
% 		($(12-last)+(2.5em,-1.6em)$)
% 		}, draw] {};	
		
\end{tikzpicture}
\caption{Visualization of sub-phase $\ell$ of the decoding stage} 
\label{fig:decoding-stage}
\end{figure}
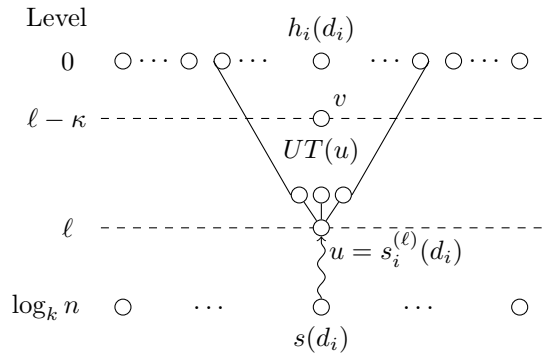
%%%%%%%%%%%%%%%%%

Table~\ref{tab:glossary} provides an overview of the variables and terms used in this work and their meanings.

\def\arraystretch{1.5}
\begin{table}
\centering
\begin{tabular}{rc|cl}
$\begin{array}{l}\mbox{Variables}\\\mbox{and Terms}\end{array}$ &&& Meaning and Notes\\\hline
$n$ &&& number of servers in the system\\
$m$ &&& size of universe of all possible keys, polynomial in $n$\\
$c$ &&& set to value $\geq 18\log m$ which is the number of pieces into which \\
&&& each data item is split before encoding it with RS codes\\
$k$ &&& set to $O(\log n)$, system uses $k$-ary butterfly as underlying topology\\
$\gamma$ &&& set to $1/72$ and used in the term $\gamma n^{1/\log\log n}$ that denotes the \\
&&& maximum number of crashed servers allowed \\
$p$ &&& positive constant in term $p\log n$ which is the length of an address\\
$\Lambda$ &&& set to $p\log n$ which is the length of an address\\
$\gamma n^{1/\log\log n}$ &&& upper bound for the number of crashed servers the system can tolerate\\
$\log_k n$ &&& depth of the underlying $k$-ary butterfly\\
$c/3$ &&& number of pieces of a data item needed to recover that data item\\
$\beta c k^2$ &&& maximum congestion at each intact server in each round\\
&&& of the decoding stage\\
$5c/6$ &&& number of pieces of a requested data item to proceed with \\
&&& in the decodings tage\\
\\
\end{tabular}
\caption{Variables and terms used in this work and their meanings}
\label{tab:glossary}

\end{table}
\def\arraystretch{1}

\end{document}